\documentclass[journal]{IEEEtran}
\usepackage[utf8]{inputenc}

\usepackage{booktabs, multirow, caption}
\usepackage{siunitx}
\usepackage{lettrine}
\usepackage{hhline}
\usepackage{adjustbox}
\usepackage{graphicx}
\usepackage[caption=false,font=footnotesize]{subfig}
\usepackage{xcolor}
\usepackage[dvipsnames]{xcolor}
\usepackage{float}
\graphicspath{images/}
\usepackage{mathrsfs}
\usepackage{caption}
\usepackage{cite}
\usepackage{subcaption}
\usepackage{comment}
\usepackage[linesnumbered,ruled,vlined]{algorithm2e}
\usepackage{adjustbox}    
\usepackage{pdfpages}
\usepackage{tabularx}
\usepackage{tikz}
\usetikzlibrary{shapes.arrows, positioning}
\usepackage{prooftrees}
\usepackage{amsthm}
\newtheorem{claim}{Claim}
\usepackage{algpseudocode}
\usepackage{array}
\renewcommand{\arraystretch}{1.2}

\renewcommand{\arraystretch}{1.0}

\usepackage[shortlabels]{enumitem}
\usepackage{booktabs}
\usepackage{multirow}

\usepackage{siunitx}
\usepackage{threeparttable}
\usepackage{hyperref}

\usepackage{amsmath}
\usepackage{amssymb}
\usepackage{stackrel}
\usepackage{mathtools}
\usepackage{pgfplots}
\pgfplotsset{compat=1.17}
\usepackage{amsmath}
\usepackage{lipsum}
 
\usepackage{amsthm}
\newtheorem{theorem}{Theorem}

\newtheorem{lemma}{Lemma}
\theoremstyle{definition}
\newtheorem{definition}{Definition}[section]

\newcommand{\negl}{\mathsf{negl}}

\newtheoremstyle{game}%
  {3pt}
  {3pt}
  {\itshape}
  {}
  {\bfseries}
  {.}
  { }
  {}
\theoremstyle{game}
\newtheorem{game}{Game}[section]
\usepackage{setspace}
\usepackage[a4paper, total={7.3in, 10.2in}]{geometry}
\renewcommand{\arraystretch}{1.5}

\usepackage[T1]{fontenc}
\usepackage{microtype}
\usepackage{booktabs, multirow, tabularx}
\newcolumntype{Y}{>{\centering\arraybackslash}X}
\hypersetup{
    colorlinks=true,
    linkcolor=blue,
    urlcolor=blue,
    citecolor=blue,
    pdftitle={Cover-Parameterised Multichannel Hybrid Steganography: Compositional Security, Detectability, and Robustness},
    pdfauthor={Obinna Omego and Michal Bosy},
    pdfsubject={Multichannel steganography and secure covert communication},
    pdfkeywords={multichannel steganography, hybrid steganography, compositional security, QIM, steganalysis}
}

\ifCLASSINFOpdf
\else
\fi

\begin{document}
%
\title{Cover-Parameterised Multichannel Hybrid Steganography: Compositional Security, Detectability, and Robustness}

%
%
%
%

\author{Obinna Omego, 
        Micha\l{} Bosy
\thanks{Obinna Omego and Micha\l{} Bosy are with the School of Computer Science and Mathematics, Faculty of Engineering, Computing and the Environment, Kingston University London.}
\thanks{Obinna Omego e-mail: a.omego@kingston.ac.uk; omegoobinna@gmail.com}
\thanks{Micha\l{} Bosy e-mail: m.bosy@kingston.ac.uk.}
\thanks{This manuscript is a preprint uploaded to arXiv.}
\thanks{\noindent\textbf{Supplementary material:}
Detailed constructions, security proofs, additional experimental results, and ablation analyses are provided in \texttt{Supplementary\_Material.pdf}, available through the ancillary-file links on this article's arXiv abstract page.}}

\markboth{\parbox{\textwidth}{This manuscript is a preprint uploaded to arXiv.}}{\parbox{\textwidth}{This manuscript is a preprint uploaded to arXiv.}}%

\maketitle


\begin{abstract}
Secure covert communication across multiple observable channels requires concealing both the transmitted objects and the relationships among them while resisting active manipulation. This paper introduces a cover-parameterised multichannel hybrid steganographic framework that combines message-independent cover synthesis with adaptive cover modification. Synthesised cover-parameter objects condition a keyed QIM-style mask, and the resulting masked payload is embedded into an existing image using QIM-Fused-CF, which integrates fused S-UNIWARD/MiPOD distortion costs, complexity-aware region refinement, SLIC-guided constraints, and syndrome-trellis coding. The protocol distributes each authenticated epoch across three channels and incorporates freshness verification, bounded scheduling, synchronisation, and re-synchronisation. We formalise cover-parameter indistinguishability (\textsf{CP-IND}) and unlinkability (\textsf{CP-UNL}), multichannel trace indistinguishability
(\textsf{IND-STEGO-MC}), and an active $\mathsf{MC\mbox{-}ATTACK}$ model covering message recovery, replay, and authenticated substitution. The resulting bounds separate masking, embedding, scheduling, authentication, and receiver-state contributions. Experiments on 10,000 BOSSBase
images achieve zero bit-error rate for all valid embeddings and structural similarity above $0.9995$ across payloads of $0.10$--$0.40$~bpp. At $0.10$ and $0.20$~bpp, SRM+EC, Ye-Net, and Yedroudj-Net produce AUC values of $0.5017$--$0.5444$ and $0.5424$--$0.5678$, respectively, whereas detectability increases substantially at $0.40$~bpp. These results identify a practical low-to-moderate-payload
operating region and demonstrate that secure multichannel steganography requires the joint design of synthesis, masking, embedding, scheduling, authentication, and receiver state.
\end{abstract}


%
\IEEEpeerreviewmaketitle

%
%
%
%

 
\section{Introduction}
\label{sec:introduction}

\IEEEPARstart{S}{teganography} seeks to conceal not only the content of a communication but also its existence. This objective has become increasingly demanding because a modern warden may combine statistical residual features, ensemble classifiers, calibrated detectors, and deep neural steganalysis to identify weak embedding traces across large collections of media \cite{Cachin1998, @hopper2009provably,kombrink2024image,Kodovsky2012Ensemble,Qiao2021AdaptiveSB,Hu2023ImageSA,Zhao2023CalibrationbasedSF}. Consequently, visual fidelity or low reconstruction error alone is insufficient: a practical system must also control distributional detectability, preserve extraction correctness, resist active manipulation, and prevent an observer from linking otherwise innocuous objects transmitted over different channels.

Most steganographic constructions follow one of three principles. \emph{Cover modification} (CMO) embeds information by altering an existing object, often through content-adaptive distortion minimisation \cite{Fridrich2009,sedighi2015content}. CMO offers practical payload capacity, but every modification may contribute to an embedding footprint. \emph{Cover selection} (CSE) avoids modifying an object by selecting an existing cover that represents the message; this reduces direct distortion but is constrained by the available cover set and generally offers lower effective capacity \cite{Hu2017Coverless, Fridrich2009}. \emph{Cover synthesis} (CSY) generates a new object whose construction depends on the secret. Learned CMO systems such as SteganoGAN optimise image quality and payload jointly \cite{zhang2019steganogan}, while generative and coverless approaches increasingly employ diffusion or autoregressive models to improve image diversity, controllability, and robust recovery \cite{yu2023cross, yang2024diffstega, xu2025mddm,qi2025provably,yuan2026provably}. These advances substantially strengthen synthesis-based steganography, including several training-free and provably secure designs. Nevertheless, they primarily protect or generate the final media object; they do not by themselves address cross-channel content linkage, authenticated epoch alignment, or the timing footprint of a multichannel transport protocol.

Prior multichannel and multi-carrier approaches distribute information across social-network resources, multiple covers, or multiple receivers \cite{Ion2013,Beato2014,wang2020multichannel}. Such distribution can improve capacity, access separation, or resilience to the loss of an individual carrier. However, a full-observation warden can analyse all channels jointly and exploit correlations in content, public metadata, object sizes, ordering, and transmission times. The unresolved problem is therefore not merely how to divide a payload among several conduits, but how to compose: \emph{(i)} benign cover-parameter generation, \emph{(ii)} keyed and cover-conditioned masking, \emph{(iii)} adaptive embedding into a conventional cover, \emph{(iv)} cross-channel opacity and unlinkability, \emph{(v)} authenticated synchronisation and re-synchronisation, and \emph{(vi)} admissible timing shaping, under one explicit adversary model and security bound. To the best of our knowledge, existing CMO, CSE, CSY, and multichannel schemes do not jointly provide these properties. Table~\ref{tab:rw-compare} summarises the principal distinctions between representative steganographic method families and the framework proposed in this work.

To address the limitations in
prior studies, this paper introduces a cover-parameterised hybrid steganographic framework that combines the complementary strengths of cover synthesis and cover modification within an authenticated multichannel communication architecture. Synthesised objects are used to provide benign contextual parameters for masking, while the concealed payload is embedded adaptively into an existing cover. The resulting objects are distributed across separate communication channels and protected through coordinated epoch binding, freshness verification, integrity checking, bounded
scheduling, and re-synchronisation. The framework is designed for a strong adversarial setting in which all channels may be observed, correlated, delayed, replayed, or modified. Its security is therefore
based on the combined indistinguishability and unlinkability of the transmitted objects, rather than on the secrecy or non-collusion of
any individual channel.

\par
\noindent\textbf{\textit{Research contributions.}}
The principal contributions are as follows:
\begin{enumerate}
    \item \emph{Cover-parameterised hybrid model:} We formalise $\mathcal{S}_{\mathsf{CPM}}$ as a six-algorithm construction that separates benign parameter synthesis, cover-conditioned QIM masking, adaptive embedding, decoding, and unmasking. This distinction clarifies that the synthesised component conditions the hidden representation, whereas the final stego-object is produced through controlled cover modification.

    \item \emph{Authenticated multichannel protocol:} We introduce a three-channel protocol with content-independent epoch binding, fresh nonces, MAC-based integrity, an admissible scheduler class $\Sigma$, and a sliding-window synchronisation and re-synchronisation algorithm with explicit complexity and robustness guarantees.

    \item \emph{Compositional security framework:} We define cover-parameter indistinguishability (\textsf{CP-IND}), cover-parameter unlinkability (\textsf{CP-UNL}), and multichannel steganographic indistinguishability (\textsf{IND-STEGO-MC}). The resulting theorem decomposes a PPT warden's advantage into cover-parameter, masking, embedding, and scheduling contributions, while the active threat analysis separately treats message recovery, replay, cross-channel substitution, and forgery under the stated assumptions.

    \item \emph{Concrete implementation and evaluation:} We instantiate the model as \emph{QIM-Fused-CF}, combining a fused S-UNIWARD/MiPOD distortion map, complex-feature-aware refinement, SLIC-guided region constraints, direction-adaptive costs, and STC-based embedding. The evaluation covers cover-stream naturalness, payload fidelity, SRM+EC and CNN steganalysis, linked-versus-shuffled multichannel wardens, masking ablations, scheduling-footprint measurements, robustness and re-synchronisation, runtime, memory, and communication overhead. The results expose the intended security--capacity trade-off: low and moderate payloads offer the strongest concealment, whereas the high-payload operating point produces a more detectable embedding footprint.
\end{enumerate}

The remainder of the paper presents the formal model and QIM instantiation, the authenticated multichannel protocol, the compositional and active-adversary analyses, the experimental evaluation, and the concluding discussion.

\begin{table*}[!t]
\centering
\caption{Representative comparison of steganographic method
families most relevant to the proposed cover-parameterised
multichannel model. The comparison concerns architectural and
security scope.}
\label{tab:rw-compare}

\scriptsize
\setlength{\tabcolsep}{3.0pt}
\renewcommand{\arraystretch}{1.18}

\begin{tabularx}{\textwidth}{
@{}
p{1.55cm}
p{2.55cm}
p{2.55cm}
p{1.65cm}
p{2.55cm}
p{2.45cm}
Y
@{}}

\toprule
\textbf{Method family} &
\textbf{Representative work(s)} &
\textbf{Cover/stego-object construction} &
\textbf{Carrier or channel model} &
\textbf{Security scope} &
\textbf{Integrity and synchronisation} &
\textbf{Cross-channel linkage and timing} \\
\midrule

Adaptive CMO &
S-UNIWARD \cite{holub2014universal};
MiPOD \cite{sedighi2015content} &
Modifies an existing image using content-adaptive distortion
costs and coded embedding. &
One modified image; single steganographic channel. &
Primarily empirical resistance to steganalysis; no end-to-end
multichannel protocol proof. &
Not part of the embedding construction. &
Not applicable to the single-object model; no joint metadata or
scheduling analysis. \\

\addlinespace
Learned CMO / cover-conditioned GAN embedding &
SteganoGAN \cite{zhang2019steganogan} &
A neural encoder maps an existing cover image and a binary
message to a stego-image. &
One modified image; single channel. &
Empirical image quality, recovery, capacity, and detector
evaluation. &
No authenticated epoch binding or re-synchronisation mechanism. &
No multichannel unlinkability or transport-timing model. \\

\addlinespace
Diffusion-based CSY/coverless &
CRoSS \cite{yu2023cross};
DiffStega \cite{yang2024diffstega};
MDDM \cite{xu2025mddm} &
Generates a container image through diffusion inversion,
message-driven noise, prompts, or a password-dependent reference
image. &
One delivered generated image; auxiliary secret material may be
pre-shared. &
Mainly empirical or construction-specific security and
robustness evaluation; no common protocol-level composition
theorem. &
Some methods support robust extraction, but not active integrity,
nonce freshness, or epoch re-synchronisation. &
No explicit three-channel content-linkability or
admissible-scheduler analysis. \\

\addlinespace
Provably secure generative steganography &
StegaStyleGAN \cite{Su2024StegaStyleGAN};
Qi \emph{et al.} \cite{qi2025provably};
Yuan \emph{et al.} \cite{yuan2026provably} &
Uses distribution-preserving mappings with GAN,
autoregressive, or latent-diffusion generation; Qi
\emph{et al.} additionally use cross-modal text for error
correction. &
Predominantly one generated stego-image; auxiliary recovery
modality in the cross-modal scheme. &
Formal media-level undetectability under the stated generator
and distribution assumptions; robustness is method dependent. &
Robust recovery is addressed in some schemes, but not an
integrated MAC-, nonce-, and epoch-based multichannel protocol. &
No explicit decomposition of content linkage and transport
scheduling across independently observed channels. \\

\addlinespace
Platform-mediated covert sharing &
Beato \emph{et al.} \cite{Ion2013,Beato2014} &
Encodes protected information in ordinary objects exchanged
through existing web or online-social-network platforms. &
Platform-hosted objects and communication sessions. &
Includes confidentiality and, in the OSN case, formal
definitions of undetectable communication. &
Does not provide the present three-frame epoch-binding and
sliding-window re-synchronisation construction. &
Does not decompose a full-observation warden's advantage into
cover-parameter unlinkability and an admissible scheduling term. \\

\addlinespace
Multichannel image steganography &
Wang \emph{et al.} \cite{wang2020multichannel} &
Embeds different secret messages in the same cover image for
separable or sequential recovery by multiple receivers. &
Multiple logical receiver channels within one physical image
carrier. &
Modern image-steganography security evaluation; no
compositional full-protocol advantage bound. &
No integrated active integrity, freshness, or cross-carrier
re-synchronisation layer. &
Does not model independently transported objects, joint content
linkage, or transmission-time shaping. \\

\addlinespace
\textbf{Proposed CPM model} &
\textbf{This work} &
Synthesises two benign cover-parameter objects, derives a
cover-conditioned QIM mask, and adaptively embeds the masked
payload into an existing cover. &
Three independently observable channels carrying
$(\gamma_1,\gamma_2,s)$. &
\textsf{CP-IND}, \textsf{CP-UNL}, and
\textsf{IND-STEGO-MC}, with a compositional bound over
cover-parameter, masking, embedding, and scheduling terms. &
Content-independent epoch headers, fresh nonces, MAC
verification, bounded windows, and sliding-window
re-synchronisation. &
Explicit linked-versus-shuffled probes and an admissible
scheduler $\sigma^\star$ with empirical KS, MMD, and
timing-warden evaluation. \\

\bottomrule
\end{tabularx}
\end{table*}

\begin{figure*}[ht]
    \centering
    \includegraphics[width=1.0\linewidth]{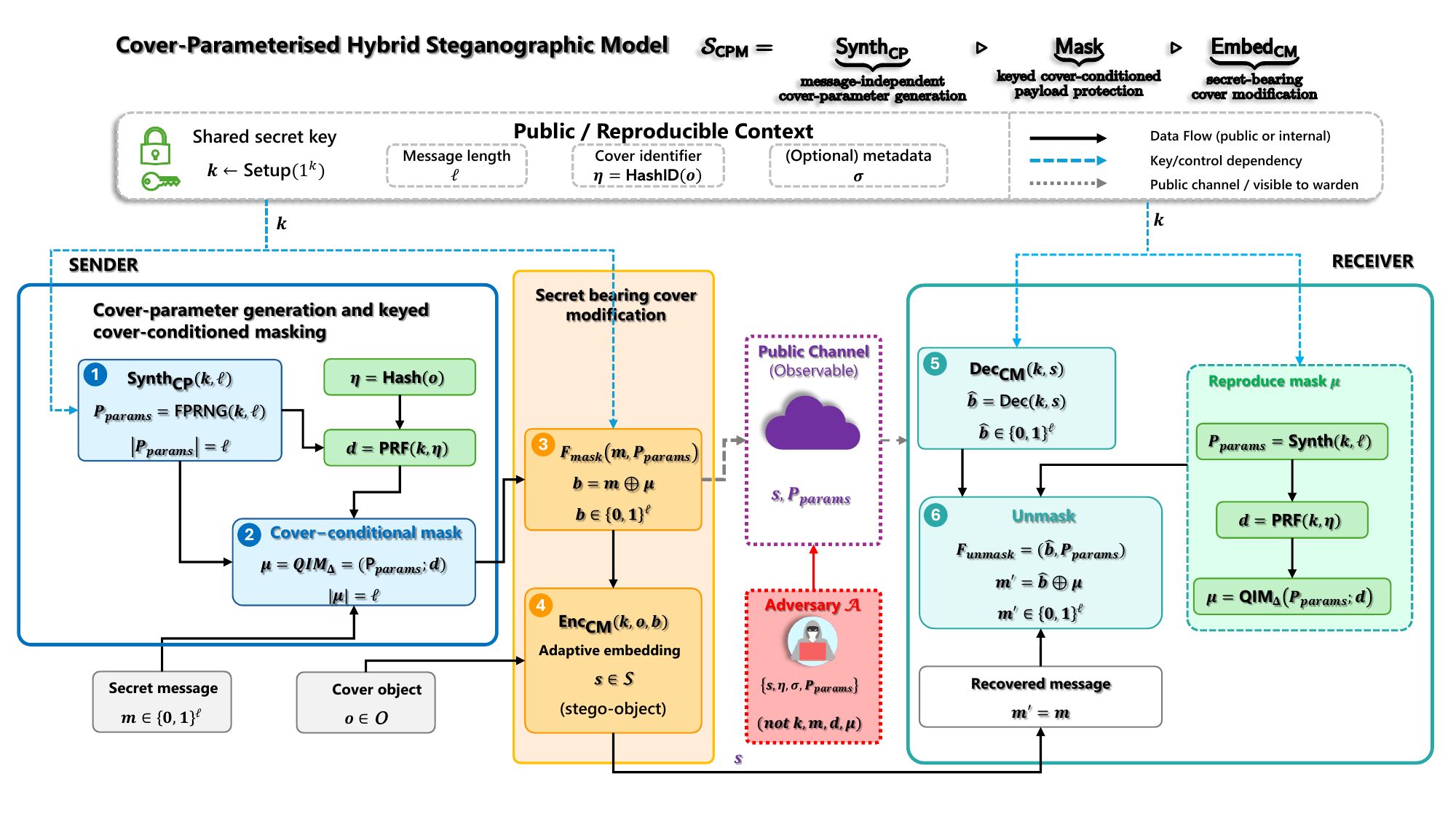}
    \caption{The Cover-Parameterised Multichannel Hybrid Steganographic Model synthesises a key-driven parameter stream $P_{\mathsf{params}}$, derives a cover-linked mask, masks the secret, and embeds only the masked payload into the cover object $o$. The receiver reproduces the same mask using the shared key and public context, and recovers the original message $m$.}
    \label{fig:hybrid-entropy-stego-model}
\end{figure*}

\section{The Proposed Steganographic System Model}
\label{Hybrid_stego_model}
This section introduces a \emph{Cover-Parameterised Multichannel Hybrid Steganographic Model}, denoted by
$\mathcal{S}_{\mathsf{CPM}}$. The model does not use conventional CSY to generate a secret-bearing cover; instead, it employs message-independent synthesis to produce cover-parameter objects suggestive of natural content, that conditions a keyed masking process through a cover-conditioned dither. The resulting masked payload is then embedded into an existing cover using CMO. This synthesis--modification architecture is refered to as \emph{Steganography by Cover-Parameterised Modification (CPM)}.Algorithm~\ref{alg:HybridStegoModel} outlines the operational flow of $\mathcal{S}_{\mathsf{CPM}}$.

\subsection{Overview and Rationale}

The motivation for CPM arises from the complementary properties of synthesis and cover modification. Conventional CMO provides a practical mechanism for embedding payloads into existing covers, but the resulting modifications may introduce detectable statistical artefacts. Conventional CSY avoids modifying an existing cover by generating the secret-bearing object directly, but may incur substantial generation cost, domain dependence, or
capacity constraints. CPM adopts a different hybrid architecture: message-independent synthesis produces innocuous cover-parameter objects, while the actual concealed payload is carried by an adaptively modified cover. Thus, the synthesis component provides
auxiliary conditioning information rather than a secret-bearing cover.

The resulting separation creates distinct security roles for the cover-parameter, masking, embedding, and multichannel transport layers (See \ref{Hybrid_Entropy_Steganographic_Communication_Protocol} for details). Cover parameters must remain benign and difficult to link to the stego-object; the keyed cover-conditioned mask protects the payload; the embedding stage controls the modification footprint; and the multichannel scheduler limits correlations introduced by
transport. These leakage sources are formalised separately and combined in the security analysis of
Theorem~\ref{thm:composition}. Figure~\ref{fig:hybrid-entropy-stego-model} illustrates the resulting separation between cover-parameter synthesis, cover-conditioned masking, and final cover modification.

Section~\ref{Hybrid_Entropy_Steganographic_Communication_Protocol} will later illustrate how these steps fit into a communication protocol and how they seamlessly interact with integrity checks and multichannel architectures.

\subsection{Formal Definition of the Cover-Parameterised Model}\label{hybrid_model}

\begin{definition}[Cover-Parameterised  Model]
\label{definition_Hybrid_stego}
\emph{A Cover-Parameterised Multichannel Hybrid Steganographic Model, denoted as $\mathcal{S}_{\mathsf{CPM}}$,  is a hybrid architecture that couples message-independent cover-parameter synthesis with keyed masking and cover-modification embedding. Specifically,}
\begin{equation*}
\label{eq:cpm-architecture}
\mathcal{S}_{\mathsf{CPM}}
=
\underbrace{\mathsf{Synth}_{\mathsf{CP}}}_{\substack{\text{message-independent}\\\text{cover-parameter generation}}}
\;\triangleright\;
\underbrace{\mathsf{Mask}}_{\substack{\text{keyed cover-conditioned}\\\text{payload protection}}}
\;\triangleright\;
\underbrace{\mathsf{Embed}_{\mathsf{CM}}}_{\substack{\text{secret-bearing}\\\text{cover modification}}},
\end{equation*}
\emph{where $\triangleright$ denotes sequential functional composition and consists of the following six efficient algorithms, each operating in polynomial time with respect to a security parameter~$\lambda$:}

\begin{enumerate}[label=(\Alph*)]
    \item 
    $\boldsymbol{\mathsf{Setup}}(\lambda)$: 
    A probabilistic algorithm that takes as input a security parameter $\lambda$ and outputs a \emph{stego-key} $k_{\mathsf{stego}}\in \mathcal{K}$.

    \item 
    $\boldsymbol{\mathsf{Synth}_{\mathsf{CP}}}(k, \ell)$:
    A \emph{cover-generation} algorithm that takes a stego-key $k$ and a message length~$\ell$. It produces a \emph{cover parameter} $P_{\mathsf{params}} \in \{0,1\}^\ell$ by invoking a secure pseudorandom process. Formally,
    \begin{equation}\label{eq:Pparams}
       P_{\mathsf{params}} \;=\;
       F_{\mathsf{PRNG}}(k,\ell)
       \quad\text{with}\quad
       |P_{\mathsf{params}}|= \ell.
    \end{equation}

    \item 
    $\boldsymbol{F_{\mathsf{mask}}}(m, P_{\mathsf{params}};\eta)$:
    A deterministic, \emph{cover-conditioned} masking algorithm. Let $\eta\in\{0,1\}^{*}$ be a public tag deterministically derived from the chosen cover (e.g., $\eta=\mathsf{HashID}(o)$ or a session header). Define a dither sequence 
    \[
        d \;\gets\; \mathsf{PRF}(k,\eta) \in \{0,1\}^{\ell},
    \]
    and a binary QIM-style mask 
    \[
        \mu \;\gets\; \mathrm{QIM}_\Delta\!\big(P_{\mathsf{params}}; d\big)\in\{0,1\}^{\ell},
    \]
    where $\mathrm{QIM}_\Delta(\cdot; d)$ denotes a scalar, two-coset (binary) quantisation-index mapping with dither $d$ and step $\Delta>0$ (per \emph{quantisation index modulation}). The masked intermediary is then
    \begin{equation}\label{eq:qim_mask}
        b
        \;=\;
        F_{\mathsf{mask}}(m, P_{\mathsf{params}};\eta)
        \;=\;
        m \;\oplus\; \mu,
    \end{equation}
    so that the whitening mask $\mu$ is \emph{keyed and cover-conditioned} via $(k,\eta)$.

    \item 
    $\boldsymbol{\mathsf{Enc}_{\mathsf{CM}}}(k,\,o,\,b)$:
    A probabilistic \emph{embedding} function that takes the stego-key $k\in\mathcal{K}$, a cover object $o\in\mathcal{O}$ (e.g., an image or text), and the masked message $b\in\{0,1\}^\ell$. It outputs a \emph{stego-object} $s \in \mathcal{S}$. Depending on the application, this step may alter selected bits of $o$ (cover modification) or embed $b$ into artificially generated or partially synthesized content.

    \item 
    $\boldsymbol{\mathsf{Dec}_{\mathsf{CM}}}(k,\,s)$:
    A deterministic \emph{decoding} function that takes the stego-key $k$ and a stego-object $s$, and returns the intermediary $b$. Symbolically, $b \leftarrow \mathsf{Dec}(k,s)$.

    \item 
    $\boldsymbol{F_{\mathsf{unmask}}}(b, P_{\mathsf{params}};\eta)$:
    An unmasking algorithm which recovers the original secret $m$ from $b$ using the same QIM-style mask $\mu$ induced by $(k,\eta)$:
    \begin{equation}\label{eq:qim_unmask}
        m=F_{\mathsf{unmask}}\bigl(b,P_{\mathsf{params}};\eta\bigr)
        = b\, \oplus \, \mathrm{QIM}_\Delta\!\big(P_{\mathsf{params}};\mathsf{PRF}(k,\eta)\big).
    \end{equation}
\end{enumerate}

\emph{We require that for all $m$ with $|m|=\ell<\mathsf{Pol}(|o|)$ and for every key~$k$, the following correctness property holds:}
\begin{equation*}
m
=F_{\mathsf{unmask}}\!\Bigl(
    \mathsf{Dec}\bigl(k,
    \mathsf{Enc}\!\bigl(k,
    o,
    F_{\mathsf{mask}}\!(m,P_{\mathsf{params}};\eta)
    \bigr)\bigr),
    P_{\mathsf{params}};\eta
\Bigr),
\end{equation*}
\emph{where $P_{\mathsf{params}}=\mathsf{Synth}(k,\ell)$, $\eta$ is a public, cover-linked tag (e.g., $\eta=\mathsf{HashID}(o)$), and $\mathsf{Pol}(\cdot)$ is a polynomial function indicating permissible payload sizes.}
\end{definition}

\begin{algorithm}[t]
\small
\caption{$\mathcal{S}_{\mathsf{CPM}}$ with Cover–Conditioned QIM Mask}
\label{alg:HybridStegoModel}
\KwIn{$\lambda$, secret $m$, cover object $O$}
\KwOut{$S$ (stego), $m'$ (recovered)}

\BlankLine
\noindent\textbf{1. Setup}\;
\quad $k \gets \mathsf{KeyGen}(\lambda)$,\ \ $\ell \gets |m|$,\ \ $P_{\mathsf{params}} \gets \mathsf{PRNG}(k,\ell)$\;

\BlankLine
\noindent\textbf{2. Cover–conditioned masking (QIM dither)}\;
\quad \tcp{Partition $O$ into $N$ host samples (or coefficients) $x[1{:}N]$ aligned to $\ell$ (via padding or coding).}
\quad $\Delta \gets \mathsf{StepSizeSelect}(O,\ell)$\; 
\quad \For{$i=1$ \KwTo $N$}{
  \tcp{Per-sample dither from key + local cover context (e.g., a small patch or feature)}
  $\delta[i] \gets \mathsf{H}\!\big(k \,\|\, \mathsf{Feat}(O,i)\big) \bmod \Delta$\;
  \tcp{Bit to embed after whitening with $P_{\mathsf{params}}$}
  $b[i] \gets m[i] \oplus P_{\mathsf{params}}[i]$\;
}
\BlankLine

\noindent\textbf{3. Encoding (Embedding)}\;
\quad $O_{\mathsf{repr}} \gets \mathsf{ToHostSamples}(O)$\;
\For{$i=1$ \KwTo $N$}{
  \tcp{QIM with two cosets: $\Lambda_0=\Delta\mathbb{Z}$, $\Lambda_1=\Delta\mathbb{Z}+\Delta/2$ (shifted by $\delta[i]$)}
  $y[i] \gets \mathsf{QIM\_Embed}\big(x[i],\, b[i],\, \Delta,\, \delta[i]\big)$\;
}
\quad $S \gets \mathsf{ReconstructObject}(y[1{:}N])$\;

\BlankLine
\noindent\textbf{4. Transmit}\;
\quad \textbf{send} $S$ over a channel\;

\BlankLine
\noindent\textbf{5. Decoding (Extraction)}\;
\quad $y'[1{:}N] \gets \mathsf{ToHostSamples}(S)$\;
\For{$i=1$ \KwTo $N$}{
  \tcp{Re-derive the same per-sample dither}
  $\delta[i] \gets \mathsf{H}\!\big(k \,\|\, \mathsf{Feat}(O,i)\big) \bmod \Delta$\;
  $b'[i] \gets \mathsf{QIM\_Detect}\big(y'[i],\, \Delta,\, \delta[i]\big)$\;
}

\BlankLine
\noindent\textbf{6. Unmask}\;
\quad $m'[i] \gets b'[i] \oplus P_{\mathsf{params}}[i] \quad (1\le i\le \ell)$\;

\BlankLine
\noindent\textbf{7. Output}\;
\If{$m = m'$}{
    $\mathsf{Display}$(\text{``Message successfully recovered."})\;
}
\Else{
    $\mathsf{Display}$(\text{``Error: $m \neq m'$"})\;
}
\end{algorithm}

The operator $\mathsf{Synth}_{\mathsf{CP}}$ draws on the
generation principle associated with cover synthesis, but it does not generate a secret-bearing stego-object. Its outputs are message-independent auxiliary cover-parameter objects. The secret-bearing object is produced only by
$\mathsf{Enc}_{\mathsf{CM}}$ through modification of an existing cover.

The model ensures that, once the shared key $k$ is established and a cover $o$ is chosen, the \emph{masked} payload can be embedded into $o$, yielding a stego object $s$ that is statistically indistinguishable from an innocent sample under standard steganalysis (as evaluated in later sections). See Figure~\ref{fig:hybrid-entropy-stego-model} for an overview.


Masking then combines $m$ with a \emph{cover-conditioned dither} to decorrelate the payload from observable cover statistics. Concretely, we use a QIM-style, cover-conditioned mask: the dither sequence is derived from $(k, P_{\mathsf{params}})$ and lightweight features of $o$ (e.g., per-block energy or local variance) via a keyed extractor, and is applied in the masking step (cf.\ Eq.~\eqref{eq:qim_mask}). The complete cover-conditioned masking construction is provided in Section~B-A of the Supplementary Material.

\begin{theorem}[Correctness of $\mathcal{S}_{\mathsf{CPM}}$]
\label{thm:CorrectnessHybrid}
Let $\mathcal{S}_{\mathsf{CPM}}$ be the system in Definition~\ref{definition_Hybrid_stego} with algorithms
$\{\mathsf{Setup},\;\mathsf{Synth},\;F_{\mathsf{mask}},\;\mathsf{Enc},\;\mathsf{Dec},\;F_{\mathsf{unmask}}\}$.
Let $\zeta \leftarrow \Phi(o)$ denote a deterministic cover side–information string (e.g., a feature map) that both parties can recompute consistently from the cover domain\footnote{In practice, $\Phi$ is chosen so that its output is preserved or reproducible at the receiver (e.g., invariant/robust features or a header reconstructed from $s$).}. Let the QIM-style dither be
$r \gets \mathsf{Dith}(k,\zeta)$
and define the cover–conditioned masking/unmasking as
\[
F_{\mathsf{mask}}(m,P_{\mathsf{params}},\zeta)= m \oplus r
\,\text{and}\,
F_{\mathsf{unmask}}(b,P_{\mathsf{params}},\zeta)= b \oplus r,
\]
where $r=\mathsf{Dith}(k,\zeta)$ and $P_{\mathsf{params}}=\mathsf{Synth}(k,\ell)$ as in~\eqref{eq:Pparams}.
For any valid key $k \leftarrow \mathsf{Setup}(\lambda)$, any $m$ of length $\ell$, and any cover $o$ with $|m|\le \mathsf{Pol}(|o|)$, if the stego $s=\mathsf{Enc}(k,o,b)$ is received without adversarial alteration, then
\[
F_{\mathsf{unmask}}\!\bigl(\mathsf{Dec}(k,s),\,P_{\mathsf{params}},\,\Phi(o)\bigr) \;=\; m
\]
holds with probability~$1$.
\end{theorem}

\begin{proof}
By construction $\Phi$ is deterministic, so both endpoints compute the same $\zeta=\Phi(o)$. Given $k$ and $\zeta$, the dither generator $\mathsf{Dith}$ is deterministic, hence both derive the same $r=\mathsf{Dith}(k,\zeta)$. The masking/unmasking pair
$(b,m)=(m\oplus r,\;b\oplus r)$ are mutual inverses for $r$. Since $\mathsf{Dec}(k,\cdot)$ inverts $\mathsf{Enc}(k,\cdot,\cdot)$, the receiver recovers $b$ exactly and then unmasking with the same $r$ yields $m$. Therefore, the stated equality holds with probability~$1$.
\end{proof}

Overall, the Hybrid Steganographic System Model $\mathcal{S}_{\mathsf{CPM}}$ provides a robust foundation for secure, multi-stage hiding of data. It is \emph{correct} by Theorem~\ref{thm:CorrectnessHybrid} and \emph{resilient} in practice, as later sections demonstrate via security analyses against multichannel adversaries (see Sections~\ref{sec:security}). 

\begin{table}[ht]
\centering
\caption{Description of Notation}
\label{tab:notation_description}
{\scriptsize
\begin{tabular}{|c|p{5.0cm}|}
\hline
\textbf{Symbol} & \textbf{Definition} \\ \hline
$m \in \mathcal{M}$ & Secret message; $\ell = |m|$ is its bit length. \\ \hline
$\ell = |m|$ & Length of secret message $m$. \\ \hline
$k \in \mathcal{K}$ & Stego-key (from $\mathsf{Setup}(1^\lambda)$). \\ \hline
$V_{\mathsf{pri}}$ & Master secret used to derive $k_{\text{stego}}$ and MACs. \\ \hline
$P_{\mathsf{params}} \in \{0,1\}^{\ell}$ & cover parameter from $\mathsf{Synth}(k,\ell)$. \\ \hline
$b \in \{0,1\}^{\ell}$ & Masked payload: $b = F_{\mathsf{mask}}(m,P_{\mathsf{params}}) = m \oplus P_{\mathsf{params}}$. \\ \hline
$o \in \mathcal{O}$ & Cover object; $s \in \mathcal{S}$ is the resulting stego object. \\ \hline
$\gamma_1,\gamma_2 \in \mathcal{O}'$ & Synthesized cover messages sent on $C_1,C_2$. \\ \hline
$k_{\text{stego}}$ & Stego key derived from $V_{\mathsf{pri}}$ (e.g., $k_{\text{stego}}=\mathsf{H}(V_{\mathsf{pri}}\oplus \mathrm{HMAC}(V_{\mathsf{pri}}\|\gamma_i))$). \\ \hline
$\mathrm{HMAC},\,\mathsf{H}$ & PRF/MAC and hash (random-oracle idealisation in proofs). \\ \hline
$C_1,C_2,C_3$; $\mathcal{C}$ & Channels; $\mathcal{C}=\{1,\dots,c\}$ and $c=|\mathcal{C}|$. \\ \hline
$H_e=(e,N_e)$ & Epoch header with nonce $N_e=\mathrm{HMAC}(V_{\mathsf{pri}}\|e)$. \\ \hline
$\mathsf{Sched}(\cdot;\sigma)$ & Admissible scheduler with state $\sigma\in\Sigma$. \\ \hline
$W,\tau,\Delta,h$ & Window size, per-epoch timeout, allowable arrival skew bound, heartbeat period. \\ \hline
$(n,k_{\text{ecc}},t)$ & ECC parameters (length, dimension, $t$ correctable errors/erasures). \\ \hline
$\mathcal{D}_i$ & Cover distribution on channel $i$; draws $o_i\sim\mathcal{D}_i$. \\ \hline
$\mathrm{TV}(\cdot,\cdot)$ & Total variation distance; used in indistinguishability bounds. \\ \hline
$\varepsilon_{\mathsf{param}},\varepsilon_{\mathsf{emb}},\varepsilon_{\mathsf{sched}}(\sigma)$ & Mask/param PR randomness gap; embedding footprint; scheduler distortion budget. \\ \hline
$U_\ell$ & Uniform distribution on $\{0,1\}^{\ell}$. \\ \hline
$\theta \leftarrow \mathsf{Synth}(k,\mathtt{ctx})$ & Synthesis parameter (variance/latent matching). \\ \hline
$\mathcal{A}$ & PPT adversary/warden (IND-STEGO-MC game). \\ \hline
$\text{PRNG}: \mathcal{K} \times \mathbb{N} \rightarrow \{0,1\}^\ell$ & Function generating pseudo-random strings. \\ \hline
\footnotesize{$F_{\text{mask}}: \mathcal{M} \cdot \{0,1\}^N \rightarrow \{0,1\}^\ell$} & Masking function combining $M$ and $P_{\mathsf{params}}$. \\ \hline
$\mathsf{Enc}: \{0,1\}^\ell \times \mathcal{O} \rightarrow \mathcal{S}$ & Embedding function for secret $b$ in cover $o$. \\ \hline
$\mathsf{Dec}: \mathcal{S} \rightarrow \{0,1\}^\ell$ & Extraction function retrieving $b$ from $s$. \\ \hline
{\scriptsize $F_{\text{unmask}}:\{0,1\}^\ell \cdot \{0,1\}^\ell \rightarrow \mathcal{M}$} & Unmasking function to recover $M$. \\ \hline
\end{tabular}
}
\end{table}

\section{Hybrid Steganographic Communication Protocol}
\label{Hybrid_Entropy_Steganographic_Communication_Protocol}

Building on Section~\ref{Hybrid_stego_model}, we define
$\mathcal{P}^{\mathsf{cs,cm}}_{\mathsf{hyb\text{-}stego}}
=\mathcal{P}(\mathcal{S}_{\mathsf{CPM}})$, which combines message-independent cover-parameter synthesis, keyed masking, adaptive cover modification, and authenticated multichannel transmission.

For each epoch, Amara sends the cover-parameter objects over $C_1$ and $C_2$ and the stego-object over $C_3$. The complete epoch transcript is authenticated. Ebere aligns and verifies the received frames, extracts the masked payload, regenerates the mask, and recovers $m$; stale, malformed, or unauthenticated epochs are rejected.

\subsection{Protocol Description}
\label{proto_Des}

\subsubsection{Setup and Shared Parameters}
\label{sec:ProtocolSetup}

Amara (sender) and Ebere (receiver) are pre-provisioned with a shared master secret $V_{\mathsf{pri}}$, which is protected by their trusted endpoint modules. The establishment and lifecycle management of $V_{\mathsf{pri}}$ are outside the formal threat
model; practical provisioning options are discussed in
Appendix~C of the Supplementary Material. They agree on (i) the cover format $o\!\in\!\mathcal{O}$ (e.g., image or text), (ii) three communication channels $(C_1,C_2,C_3)$ that are \emph{unsecured individually} but used in composition, and (iii) session transport knobs used later for alignment and recovery: a per-epoch header
\[
H_e \triangleq (e, N_e),\qquad N_e := \mathrm{HMAC}(V_{\mathsf{pri}}\|e),
\]
a sliding window size $W$, per-epoch timeout $\tau$, allowable inter-channel skew $\Delta$, and (optionally) an $(n,k,t)$ error/erasure-correcting code protecting the $C_3$ payload stream. Domain separation is enforced so that header HMACs and content MACs use disjoint labels.

\subsubsection{Phases of the Protocol}
\label{ProtocolPhases}

\paragraph{ Setup Phase}
Amara establishes $C_1$, $C_2$, and $C_3$. She samples fresh nonces $\mathsf{nonce}_a,\mathsf{nonce}_b$ to bind cover messages to the session, preventing repeating messages.

\paragraph{Message Generation and Transmission}
Amara prepares three messages:
\begin{itemize}
  \item The secret message $m$ of length~$\ell$ that she wishes to send.
  \item Two \emph{cover messages} $(\gamma_{1}, \gamma_{2})$ of equal length~$\ell$, each produced by
    $(\gamma_{1}, \gamma_{2})\;\gets\; \mathsf{Synth}\bigl(V_{\mathsf{pri}},\,\ell\bigr)$,
    using the parameters generated in Equation~\eqref{eq:Pparams}.  These cover messages appear statistically benign and do not hint at any embedded secret.
\end{itemize}

She concatenates $\mathsf{nonce}_a$ with $\gamma_{1}$ and $\mathsf{nonce}_b$ with $\gamma_{2}$, sending the pairs $(\mathsf{nonce}_a \| \gamma_{1})$ and $(\mathsf{nonce}_b \| \gamma_{2})$ over channels $C_1$ and $C_2$, respectively. Both messages can appear as ordinary short texts or data blocks, each with its distinct nonce appended at the start. Each prefixed by the current epoch header $H_{e}$

\paragraph{Message Masking and Encoding}
\label{SubsecMaskingEncoding}
To mask $m$, Amara derives an auxiliary secret $k_{\mathsf{stego}}$, which she will use for steganographic embedding. Algorithm~\ref{alg:MessageEncoding} demonstrates the main steps:

\begin{algorithm}[t]
\small
\caption{Message Encoding and MAC Generation (with cover-conditioned dither)}
\label{alg:MessageEncoding}
\KwIn{$\gamma_{1},\; \gamma_{2},\; m,\; V_{\mathsf{pri}},\; o$}
\KwOut{$(\mathsf{nonce}_c,\; s,\; \text{MAC})$}

\BlankLine
\noindent\textbf{1. Compute an intermediate HMAC}\;
$P \gets \mathrm{HMAC}(V_{\mathsf{pri}} \| \gamma_{i})$ \;
\BlankLine

\noindent\textbf{2. Derive the stego-key}\;
$k_{\text{stego}} \gets \mathsf{H}(\,V_{\mathsf{pri}} \,\oplus\, P)\;$ 
\BlankLine

\noindent\textbf{3. Cover-conditioned dither (QIM style)}\;
$\rho_{\mathsf{dit}} \gets \mathsf{Dither}(k_{\text{stego}},\, \psi(o))$ \tcp*[f]{keyed dither from public cover features $\psi(o)$}
\BlankLine

\noindent\textbf{4. Mask the secret}\;
$P'_{\mathsf{params}} \gets \gamma_{1} \oplus \gamma_{2} \oplus \rho_{\mathsf{dit}}$ \;
$b \gets F_{\mathsf{mask}}(m,\,P'_{\mathsf{params}}) = m \oplus \gamma_{1} \oplus \gamma_{2} \oplus \rho_{\mathsf{dit}}$ \;
\BlankLine

\noindent\textbf{5. Steganographically encode into cover object}\;
$s \gets \mathsf{Enc}\bigl(b,\, k_{\text{stego}},\, o \bigr)$\;
\BlankLine

\noindent\textbf{6. Generate a fresh nonce and a MAC}\;
$\mathsf{nonce}_c \gets \text{GenerateNonce}()$ \;
$\mathrm{MAC} \gets \mathrm{HMAC}(\mathsf{nonce}_c \,\|\, s,\; V_{\mathsf{pri}})$ \;
\BlankLine

\Return $(\mathsf{nonce}_c,\; s,\; \mathrm{MAC})$\;
\end{algorithm}

\textbf{Key steps in Algorithm~\ref{alg:MessageEncoding}:}
\begin{enumerate}
    \item \textbf{Auxiliary HMAC:} 
    She computes $P = \mathrm{HMAC}(V_{\mathsf{pri}} \| \gamma_{i})$ with one of the cover messages $\gamma_{i} \in \{\gamma_{1},\gamma_{2}\}$. This leverages the shared secret $V_{\mathsf{pri}}$ to produce an unguessable value $P$.
    \item \textbf{Stego-Key Derivation:} 
    She hashes the XOR of $V_{\mathsf{pri}}$ and $P$ to create $k_{\text{stego}}$. This key is critical for embedding operations and serves as an additional protective layer since an adversary would need both $V_{\mathsf{pri}}$ and knowledge of $\gamma_{i}$ to reconstruct $k_{\text{stego}}$.
    \item \textbf{Cover-Conditioned Dither:} 
    Compute a keyed dither term
    \begin{equation}\label{D:mask}
        \rho_{\mathsf{dit}} := \mathsf{Dither}(k_{\text{stego}},\, \psi(o))
    \end{equation}
    where $\psi(\cdot)$ is a public, public deterministic feature extractor over the cover object.

    \item \textbf{Message Masking:}
Following Equation~\eqref{eq:Pparams}, the following is set
    \begin{equation}\label{mask_a_hybrid}
  P'_{\mathsf{params}} \;:=\; \gamma_{1} \oplus \gamma_{2} \oplus \rho_{\mathsf{dit}}, \, b:=m \oplus P'_{\mathsf{params}}.
\end{equation}
Then compute the masked payload via
\begin{equation}\label{eq:qim-dither-mask}
  b \;=\; F_{\mathsf{mask}}\bigl(m,\,P'_{\mathsf{params}}\bigr)
      \;=\; m \oplus \gamma_{1} \oplus \gamma_{2} \oplus \rho_{\mathsf{dit}}.
\end{equation}

    which yields immediately
    \begin{equation}\label{eq:qim-dither-mask_a}
  b \;=\; m \;\oplus\; \gamma_{1} \;\oplus\; \gamma_{2} \;\oplus\; \rho_{\mathsf{dit}}.
\end{equation}
    Here, $P'_{\mathsf{params}}$ is defined as the XOR of both cover messages and the stego‐key. This approach keeps the XOR “whitening” spirit but ties it to the cover via $\rho_{\mathsf{dit}}$.
    
    \item \textbf{Stego-Object Generation:} 
    The function $\mathsf{Enc}(\cdot)$ embeds $b$ into an innocuous cover object $o$, such as an image or textual data, producing $s$ as in \eqref{embed}:
    \begin{equation}\label{embed}
    s =\mathsf{Enc}\bigl(b,\,k_{\mathsf{stego}},\,o\bigr).
    \end{equation}
    This final stego-object $s$ is transmitted on channel $C_3$.

    \item \textbf{Integrity Code:}
    A fresh nonce $\mathsf{nonce}_c$ and a MAC computed  ensuring that modifications to $s$ or replays of old data are detectable:
    \begin{equation}\label{MAC}
       \mathrm{MAC} \gets \mathrm{HMAC}(\mathsf{nonce}_c \,\|\, s,\; V_{\mathsf{pri}})
    \end{equation}
\end{enumerate}

\paragraph{Message Transmission}
Amara then sends $(\mathsf{nonce}_c, s, \mathrm{MAC})$ over channel $C_3$. The final transmitted components over $C_1$, $C_2$ $C_3$ respectively per epoch $e$ are:
\begin{multline}\label{eq:trans}
   \Big\{\, H_e\|(\mathsf{nonce}_a\|\gamma_{1}),\,
       H_e\|(\mathsf{nonce}_b\|\gamma_{2}),\\
       H_e\|(\mathsf{nonce}_c, s, \mathrm{MAC})\,\Big\}.
\end{multline}

Even if channels $C_1$ or $C_2$ (carrying cover messages) are compromised, the adversary would still be missing $k_{\text{stego}}$. Likewise, knowledge of $s$ on $C_3$ without $\gamma_{1}$ and $\gamma_{2}$ is insufficient to recover $m$.

\paragraph{Message Unmasking, Decoding, and Verification}
\label{protocol:decodingverification}
Upon receiving $\bigl(\gamma_{1},\gamma_{2},s,\mathsf{nonce}_c,\mathrm{MAC}\bigr)$ and the nonces $\mathsf{nonce}_a$, $\mathsf{nonce}_b$ across the three channels, Ebere performs:

\begin{enumerate}
    \item \textbf{Header \& freshness:} verifies $H_e=(e,N_e)$ via $N_e \stackrel{?}{=}\mathrm{HMAC}(V_{\mathsf{pri}}\|e)$; checks nonces $\mathsf{nonce}_a,\mathsf{nonce}_b,\mathsf{nonce}_c$ for freshness within the window $(W,\tau,\Delta)$.

    \item \textbf{MAC Verification:} 
    She recomputes 
    \begin{equation}\label{MAC Verifi}
        \mathrm{MAC}' 
       \;=\; 
       \mathrm{HMAC}(\mathsf{nonce}_c \,\|\, s,\; V_{\mathsf{pri}}),
    \end{equation}

    checking $\mathrm{MAC}' \stackrel{?}{=} \mathrm{MAC}$. A mismatch indicates tampering, in which case Ebere aborts.

    \item \textbf{Stego-Key Reconstruction:}
    \[
       P' = \mathrm{HMAC}(V_{\mathsf{pri}} \,\|\, \gamma_{i}),\qquad
       k'_{\text{stego}} = \mathsf{H}(V_{\mathsf{pri}} \,\oplus\, P').
    \]
    Compute the receiver’s dither deterministically from the received cover features:
    $\rho'_{\mathsf{dit}} \gets \mathsf{Dither}(k'_{\text{stego}},\, \psi(o^\star)),$
    where $\psi(o^\star)$ uses the agreed public extractor on the reference cover $o$ (or the canonical cover features derivable from $s$ when $o$ is implicit.
    
    \item \textbf{Extraction and unmasking.}
    Recover $b'=\mathsf{Dec}(k'_{\text{stego}},\,s)$ and then
    \begin{multline}\label{unmasked_proto_qim}
    m' = F_{\mathsf{unmask}}\bigl(b',\,P'_{\mathsf{params}}\bigr)
          = b' \oplus \bigl(\gamma_{1} \oplus \gamma_{2} \oplus \rho'_{\mathsf{dit}}\bigr).
    \end{multline}

    Hence $b'=\mathsf{Dec}\bigl(k'_{\mathsf{stego}},\,s\bigr),
      \,
      m'=b'\oplus \gamma_{1} \oplus \gamma_{2}\oplus k'_{\mathsf{stego}}$.
\end{enumerate}
The function $\mathsf{Dither}(k_{\text{stego}},\psi(\cdot))$ and the extractor $\psi(\cdot)$ are fixed system components; see Appendix~B of the Supplementary Material for the deterministic construction and reproducibility constraints.
The process defined from Equations \eqref{D:mask}-\eqref{eq:qim-dither-mask_a} preserves the whitening role of masking while tying the mask to cover structure. Importantly, the dither is not claimed to eliminate detectability, but rather to ensure that the hybrid masking layer does not increase the detectability of the $C_3$ stego-object while providing additional separation between parameter generation and embedding.

The multichannel setting is also central to the model. By distributing cover parameters and the stego-object across distinct channels, the construction supports a stronger operational objective than ordinary single-cover embedding.

\noindent\textit{Bridge to synchronisation.} Epoch headers $H_e$ and the sliding window $(W,\tau,\Delta)$ align the three channels and enable re-synchronisation under loss and skew; see Algorithm~\ref{alg:mc-sync}.

\begin{lemma}[Protocol Correctness]
\label{lem:protocol-correctness}
Assume that the three frames of epoch $e$ are received unaltered, pass the header, freshness, and authentication checks, and that $\mathsf{Enc}_{\mathsf{CM}}/\mathsf{Dec}_{\mathsf{CM}}$ satisfy their correctness condition. If both endpoints reconstruct the same cover parameters and masking context, then the receiver recovers the original message:
\begin{equation}
m'_e=m_e .
\end{equation}
\end{lemma}

\begin{proof}
From the received cover-parameter objects, Ebere reconstructs $P_{\mathsf{params},e}$ and the same context $\eta_e$ used by Amara, and therefore regenerates the identical mask $\mu_e$. Correct decoding gives $\widehat b_e = \mathsf{Dec}_{\mathsf{CM}}
(k_{\mathsf{stego}},s_e) = b_e .$
Since
$b_e=m_e\oplus\mu_e$, unmasking yields $m'_e
= \widehat b_e\oplus\mu_e = (m_e\oplus\mu_e)\oplus\mu_e = m_e .$
\end{proof}
The complete masking-reconstruction and error decomposition are
given in Supplementary
Section~B-A of the Supplementary Material.

\subsection{Algorithm 1: Multichannel Synchronisation and Re-synchronisation}
\label{subsec:algo-sync}

The Protocol Description in Section \ref{proto_Des} specifies \emph{what} is sent over the three channels: 
(i) $(\mathsf{nonce}_a\|\gamma_{1})$ on $C_1$, 
(ii) $(\mathsf{nonce}_b\|\gamma_{2})$ on $C_2$, and 
(iii) $(\mathsf{nonce}_c, s, \mathrm{MAC})$ on $C_3$, 
where $s=\mathsf{Enc}(b,k_{\text{stego}},o)$ and $b=m\oplus\gamma_{1}\oplus\gamma_{2}\oplus k_{\text{stego}}$. 
\textbf{Algorithm~\ref{alg:mc-sync}} below specifies \emph{how} these items are \emph{scheduled, aligned, and recovered} across the three flows in the presence of delay, loss, and reordering, without interactive retransmissions. It uses short epoch headers to align the three channels and a sliding window with timeouts to re-synchronize if one channel lags. The epoch header $H_{e}=(e,N_e)$ is constant-size metadata allowed by the admissible scheduler class $\Sigma$ and is independent of $(\gamma_1,\gamma_2,s_e)$ content, which preserves CP-UNL by design (no content-derived linkage).”

Three channels $C_1,C_2,C_3$ (generalizable to $c\!\ge\!2$); admissible scheduler states $\Sigma$ with scheduler $\mathsf{Sched}(\cdot;\sigma)$ and $\sigma\!\in\!\Sigma$; epoch window $W$; per-epoch timeout $\tau$; allowable arrival skew $\Delta$; optional Error Correcting Code (ECC) $(n,k,t)$ over the masked payload stream on $C_3$; security parameter $\lambda$; master secret $V_{\mathsf{pri}}$; PRF/HMAC and hash $\mathsf{H}$ as in the protocol; per-epoch header nonce 
$
N_e \leftarrow \mathrm{HMAC}\!\big(V_{\mathsf{pri}} \,\|\, e\big)
$ 
binding the three channels to epoch $e$. We write $H_e \!\triangleq\! (e,N_e)$.

\begin{algorithm}[h!]
\footnotesize
\DontPrintSemicolon
\caption{Multichannel Synchronization and Re-synchronization (aligned to the protocol)}
\label{alg:mc-sync}
\KwIn{Protocol inputs from Sec.~\ref{proto_Des}: $V_{\mathsf{pri}}, \gamma_{1},\gamma_{2}, m, o$; scheduler state $\sigma\!\in\!\Sigma$; window $W$, timeout $\tau$, skew $\Delta$; optional ECC $(n,k,t)$.}
\KwOut{Recovered $m'$ at $\mathcal{R}$ (or $\bot$ on failure).}

\SetKwBlock{Sender}{Sender (Amara)}{end}
\SetKwBlock{Receiver}{Receiver (Ebere)}{end}

\Sender{
  \textbf{Step S1 (Per-session init).} Set epoch $e\leftarrow 0$; compute $k_{\text{stego}}$ and $b$ as in Alg.~\ref{alg:MessageEncoding} (where $b=m\oplus D_\ell(o;\gamma_i,V_{\mathsf{pri}},k_{\text{stego}})$ per Appendix~B of the Supplementary Material); if using ECC, encode $b$ into codewords $(c_1,\dots,c_B)$.\;
  \While{symbols or headers remain to ship}{
    \textbf{S2 (Epoch header).} $e\leftarrow e+1$; derive $N_e = \mathrm{HMAC}(V_{\mathsf{pri}}\|e)$; set $H_e=(e,N_e)$.\;
    \textbf{S3 (Channel $C_1$ payload).} Ship $H_e \| (\mathsf{nonce}_a\|\gamma_{1})$ on $C_1$ via $\mathsf{Sched}(C_1,\cdot;\sigma)$.\;
    \textbf{S4 (Channel $C_2$ payload).} Ship $H_e \| (\mathsf{nonce}_b\|\gamma_{2})$ on $C_2$ via $\mathsf{Sched}(C_2,\cdot;\sigma)$.\;
    \textbf{S5 (Channel $C_3$ payload).} 
      \uIf{ECC used}{
         take next codeword (or symbol bundle) for epoch $e$, embed into $o$ to produce $s_e=\mathsf{Enc}(b_e,k_{\text{stego}},o)$; 
      }\Else{
         set $s_e=\mathsf{Enc}(b_e,k_{\text{stego}},o)$ for the next segment $b_e$ of $b$;
      }
      generate fresh $\mathsf{nonce}_c$; compute $\mathrm{MAC}_e=\mathrm{HMAC}(\mathsf{nonce}_c\| s_e, V_{\mathsf{pri}})$; 
      ship $H_e \| (\mathsf{nonce}_c, s_e, \mathrm{MAC}_e)$ on $C_3$ via $\mathsf{Sched}(C_3,\cdot;\sigma)$.\;
    \textbf{S6 (Optional heartbeat).} Every $h$ epochs, transmit a header-only heartbeat $H_e$ on all channels to accelerate re-sync.\;
  }
}

\Receiver{
  \textbf{R1 (Init).} Initialize sliding epoch window $\mathcal{W}\!\leftarrow\!\{e_0,\dots,e_0\!+\!W\!-\!1\}$; per-epoch buffers $\mathcal{B}_1[e],\mathcal{B}_2[e],\mathcal{B}_3[e] \leftarrow \varnothing$.\;
  \While{any channel delivers frames before global timeout}{
    \textbf{R2 (Parse \& authenticate headers).} Upon a frame from channel $C_j$:
      parse leading header $\widetilde{H}=(\tilde{e},\tilde{N})$; 
      verify $\tilde{e}\in\mathcal{W}$ and $\tilde{N} \stackrel{?}{=} \mathrm{HMAC}(V_{\mathsf{pri}}\|\tilde{e})$.\;
      \uIf{$j=1$ and payload is $(\mathsf{nonce}_a\|\gamma_{1})$}{append to $\mathcal{B}_1[\tilde{e}]$ with arrival timestamp;}
      \uElseIf{$j=2$ and payload is $(\mathsf{nonce}_b\|\gamma_{2})$}{append to $\mathcal{B}_2[\tilde{e}]$;}
      \uElseIf{$j=3$ and payload is $(\mathsf{nonce}_c, s_{\tilde{e}}, \mathrm{MAC}_{\tilde{e}})$}{append to $\mathcal{B}_3[\tilde{e}]$;}
      otherwise discard.\;
    \textbf{R3 (Integrity on $C_3$).} For any $\tilde{e}$ where $\mathcal{B}_3[\tilde{e}]$ has $(\mathsf{nonce}_c, s_{\tilde{e}}, \mathrm{MAC}_{\tilde{e}})$, 
      recompute $\mathrm{MAC}'=\mathrm{HMAC}(\mathsf{nonce}_c\|s_{\tilde{e}},V_{\mathsf{pri}})$ and require $\mathrm{MAC}'=\mathrm{MAC}_{\tilde{e}}$.\;
    \textbf{R4 (Epoch-level readiness).} When all three buffers for some $\tilde{e}\in\mathcal{W}$ are present (within skew $\Delta$):
      \begin{enumerate}\itemsep2pt
        \item Recompute $k'_{\text{stego}}$ from $V_{\mathsf{pri}}$ and one of $\gamma_{i}\in\mathcal{B}_{1,2}[\tilde{e}]$ (as in \\ Alg.~\ref{alg:MessageEncoding}).
        \item Extract $b'_{\tilde{e}}=\mathsf{Dec}(k'_{\text{stego}}, s_{\tilde{e}})$.
        \item 
              Compute the cover-conditioned dither mask $D_\ell^{(\tilde{e})} \!=\! \mathsf{Dither}_\ell\!\bigl(o;\gamma_i,V_{\mathsf{pri}},k'_{\text{stego}}\bigr)$ (see Appendix~B of the Supplementary Material); \\set $m'_{\tilde{e}} = b'_{\tilde{e}} \oplus D_\ell^{(\tilde{e})}$.
        \item Append $m'_{\tilde{e}}$ to the message reassembly buffer in epoch \\ order; advance the window start to $\tilde{e}{+}1$.
      \end{enumerate}
    \textbf{R5 (Re-synchronization on loss/skew).} 
      For any $e'\in\mathcal{W}$ with missing items after timeout $\tau$:
      \begin{enumerate}\itemsep2pt
        \item \emph{Skew tolerance:} accept late arrivals within $\Delta$; prefer earliest \\ timestamp per channel.
        \item \emph{Epoch skip:} if some channel is still missing after $\tau$, mark $e'$ \\ skipped (erasure) and slide window by one.
        \item \emph{Heartbeat align:} upon header-only heartbeat $H_{e^\star}$, realign \\ window start to $e^\star$ if $e^\star$ falls inside the next $W$ epochs.
      \end{enumerate}
  }
  \textbf{R6 (Finalize).} Concatenate the per-epoch fragments to obtain $m'$; output $m'$ if complete, else $\bot$.
}
\end{algorithm}

The protocol has fixed the \emph{payloads} per channel. Algorithm~\ref{alg:mc-sync} adds the \emph{transport discipline}: a light epoch header $H_e=(e,N_e)$ which is derived from the same $V_{\mathsf{pri}}$, binds $(\gamma_{1},\gamma_{2})$ and $(\mathsf{nonce}_c,s,\mathrm{MAC})$ to the same epoch, and the receiver uses a sliding window with $\tau$ and $\Delta$ to re-synchronize. This cleanly connects what is transmitted, keeping the 3 flows aligned under loss/skew.

\subsection{Complexity and Robustness Analysis}
\label{subsec:complexity-robustness}


\paragraph{Computational complexity.}
Let $\ell{=}|m|$ be the message length in bits and $B$ the number of epochs emitted for $m$ (with or without ECC chunking). Define single-object embed/extract costs $T_{\mathsf{enc}}$ and $T_{\mathsf{dec}}$, HMAC time $T_{\mathsf{mac}}$, and per-bit dither/quantizer time $T_{\mathsf{dither}}$ (cover-conditioned QIM mask evaluation). Since QIM uses a constant number of quantizer lookups/rounds per symbol, $T_{\mathsf{dither}}$ is linear in $\ell$ and negligible compared to embedding/extraction in practice.

\begin{itemize}
  \item \emph{Sender.} Per epoch, constant-size header/HMAC plus one embed for $C_3$. Since $O(\cdot)$ denotes how fast computation grows,
as the input size becomes large, the dither generation for the masked stream:
  \[
    T^{(\mathrm{snd})}_{\mathrm{epoch}} = O(T_{\mathsf{enc}} + T_{\mathsf{mac}})+O(T_{\mathsf{dither}}),
  \]
  hence
  \[
    T^{(\mathrm{snd})}_{\mathrm{total}} = O\Big( B \cdot (T_{\mathsf{enc}} + T_{\mathsf{mac}})+ \ell \Big).
  \]
  If ECC is used, add $O(T_{\mathrm{ECC_{enc}}})$ per $k$-symbol block (linear in $\ell$). $T^{(\mathrm{snd})}_{\mathrm{epoch}}$ denotes the  sender’s computational cost during one epoch.
\item \emph{Receiver.} Per epoch, one MAC check, one extract for $C_3$, and the same dither evaluation (to invert the QIM mask):
   \[
    T^{(\mathrm{rev})}_{\mathrm{epoch}} = O(T_{\mathsf{dec}} + T_{\mathsf{mac}})+O(T_{\mathsf{dither}}),
  \]
  hence
  \[
    T^{(\mathrm{rev})}_{\mathrm{total}} = O\Big( B \cdot (T_{\mathsf{dec}} + T_{\mathsf{mac}})+ \ell \Big).
  \]
  ECC decoding contributes $O(T_{\mathrm{ECC\_dec}})$ per successfully accumulated codeword. $T^{(\mathrm{rev})}_{\mathrm{epoch}}$ denotes the receiver’s computational cost during one epoch.
\end{itemize}
All remaining book-keeping (window slides, header checks, timestamps) is $O(1)$ per arrival, hence $O(B)$ total under bounded per-epoch arrivals.

\paragraph{Communication overhead.}
Each epoch conveys three small headers $H_e$ (one per channel) plus two short cover strings on $C_1,C_2$ and one stego object on $C_3$. Let $|H|$ be the bytes for $(e,N_e)$ and $|\mathrm{MAC}|$ the MAC length. The per-epoch overhead (beyond payload) is
\[
  \underbrace{3|H|}_{\text{headers}} + \underbrace{|\mathsf{nonce}_a| + |\mathsf{nonce}_b| + |\mathsf{nonce}_c|}_{\text{freshness}} + \underbrace{|\mathrm{MAC}|}_{\text{integrity}} ,
\]
which is constant (tens of bytes) and independent of $\ell$. Heartbeats every $h$ epochs add $O\!\big(\tfrac{B}{h}\cdot 3|H|\big)$ bytes.

\paragraph{Memory footprint.}
The receiver buffers at most $W$ in-flight epochs per channel. With at most one retained frame per channel per epoch (earliest arrival wins within $\Delta$),
\[
  \text{Space} = O\!\big(W\cdot c\cdot |H|\big) \;+\; O\!\big(W\cdot |s_e|\big) \;+\; O\!\big(W\cdot |\gamma_i|\big),
\]
which is linear in $W$ and dominated by the stego segment size $|s_e|$ on $C_3$.

\paragraph{Robustness to loss and skew}
Let $X_i^{(e)}\!\in\!\{0,1\}$ indicate whether the frame for epoch $e$ on channel $C_i$ arrives within the $\Delta$ skew bound and before $\tau$; write $p_i \!=\! \Pr[X_i^{(e)}\!=\!1]$. Assume epochs are independent after conditioning on network state (standard in epoch-based analyses), and define per-epoch success probability
\[
  \pi \;\triangleq\; \Pr\!\left[ X_1^{(e)}{=}1 \wedge X_2^{(e)}{=}1 \wedge X_3^{(e)}{=}1 \right]
  \;\;=\;\; p_1 p_2 p_3,
\]
or, with ECC on $C_3$, replace the $C_3$ indicator by the event “at least $n{-}t$ symbols accumulated for epoch $e$”, which raises $\pi$ accordingly.

\paragraph{Detectability note (QIM mask)}
The cover-conditioned QIM mask introduces only a linear-time pre/post step but materially improves rate–distortion behavior (and thus detectability at fixed payload) by aligning embedding perturbations to quantizer cells tied to local cover features, a classical advantage of QIM-style schemes. This adjustment does not change channel-level overheads or the window/timeout machinery above \footnote{B.~Chen and G.~W.~Wornell, ``Quantization Index Modulation: A Class of Provably Good Methods for Digital Watermarking and Information Embedding,'' \emph{IEEE Transactions on Information Theory}, 47(4):1423–1443, 2001. Available at: \url{https://sia.mit.edu/wp-content/uploads/2015/04/2001-chen-wornell-it.pdf}.}.

\begin{lemma}[Windowed re-synchronization bound]
\label{lem:resync}
Let $\pi$ denote the per-epoch probability that all items for epoch $e$ are usable at the receiver—i.e., both $(\mathsf{nonce}_a\|\gamma_{1})$ and $(\mathsf{nonce}_b\|\gamma_{2})$ arrive within skew $\Delta$ and before timeout $\tau$, and $(\mathsf{nonce}_c,s_e,\mathrm{MAC}_e)$ arrives and passes MAC verification (or the ECC threshold on $C_3$ is met). With sliding window size $W$, timeout $\tau$, skew tolerance $\Delta$, and header-only heartbeats every $h\!\le\!W$ epochs, the probability that the receiver fails to re-synchronize within $L$ subsequent epochs after a desynchronizing event is at most
\[
\Pr[\text{no re-sync in }L]\;\le\;(1-\pi)^{\lfloor L/h\rfloor+1}.
\]
In particular, the expected re-sync delay is $O(1/\pi)$ epochs (geometric waiting time), capped by $W$ in the presence of heartbeats.
\end{lemma}

\begin{proof}[Proof sketch]
Model “epoch $e$ is decodable and window-advancing” as a Bernoulli($\pi$) trial: both $C_1$ and $C_2$ deliver the prefixed $(\mathsf{nonce}\|\gamma)$ within $\Delta$, and $C_3$ delivers an authenticated $(\mathsf{nonce}_c,s_e,\mathrm{MAC}_e)$ (or sufficient ECC symbols). Independence across epochs is assumed after conditioning on network state. Heartbeats every $h$ epochs guarantee at least one alignment opportunity per $h$-block; thus, within $L$ epochs there are $n=\lfloor L/h\rfloor+1$ trials. The chance of no success in $n$ independent Bernoulli($\pi$) trials is $(1-\pi)^n$, yielding the bound. The expected waiting time to the next success is $1/\pi$ by the geometric distribution.
\end{proof}

\paragraph{Correctness under admissible scheduling.}
Let $\rho$ upper-bound the scheduler’s admissible distortion (delay/jitter/burst) per trace, i.e., $\varepsilon_{\mathsf{sched}}(\sigma)\!\le\!\rho$ for all $\sigma\!\in\!\Sigma$. Then scheduling influences only the per-epoch success probability via latency coupling in the channel terms $p_i$ (hence $\pi$) while preserving indistinguishability at budget $\rho$ (see Theorem~\ref{thm:composition}).

\paragraph{Integrity, replay, and binding across channels.}
The per-epoch header $H_e\!=\!(e,N_e)$ with $N_e=\mathrm{HMAC}(V_{\mathsf{pri}}\|e)$ binds $(\mathsf{nonce}_a\|\gamma_1)$ on $C_1$, $(\mathsf{nonce}_b\|\gamma_2)$ on $C_2$, and $(\mathsf{nonce}_c,s_e,\mathrm{MAC}_e)$ on $C_3$ to the same epoch, preventing cross-channel splicing. Fresh nonces and per-epoch $\mathrm{MAC}_e=\mathrm{HMAC}(\mathsf{nonce}_c\|s_e,V_{\mathsf{pri}})$ provide authenticity and replay resistance; stale headers fail either the HMAC check or the window timing constraints.

\paragraph{Link to indistinguishability.}
Per Section~\ref{subsec:formal} and Theorem~\ref{thm:composition}, the warden’s multi-channel advantage obeys
\begin{multline*}
  \mathsf{Adv}^{\mathsf{mc}}_{\mathcal{S}_{\mathsf{CPM}},\mathbf{\mathcal{D}}}(\mathcal{A};\lambda)
\le
\underbrace{\varepsilon_{\mathsf{dither}}}_{\text{cover-conditioned mask}}
+
\underbrace{\varepsilon_{\mathsf{emb}}}_{\text{embedding footprint}}
\\ +
\underbrace{\varepsilon_{\mathsf{sched}}(\sigma^\star)}_{\text{admissible transport shaping}}
+\negl(\lambda).  
\end{multline*}

Here $\varepsilon_{\mathsf{dither}}$ captures the indistinguishability error of the QIM-style, cover-conditioned masking (replacing the former $\varepsilon_{\mathsf{param}}$), $\varepsilon_{\mathsf{emb}}$ the encoder’s footprint, and $\varepsilon_{\mathsf{sched}}$ the synchronisation discipline under the fixed admissible scheduler class~$\Sigma$ (See Algorithm~\ref{alg:mc-sync}).

\paragraph{Parameter guidance (practical).}
Choose $(n,k,t)$ to meet a per-epoch decode target $\Pr[\text{fail}]\!\le\!\delta$ under worst-case $p_i$; pick $\tau$ as the $(1{-}\alpha)$-quantile of end-to-end delay to control late-drop rate $\alpha$; set $W \!\ge\! \lceil \Delta/\text{epoch\_duration} \rceil + t$ to absorb skew and erasures; select heartbeat period $h \!\le\! W/2$ for prompt re-alignment; and enforce scheduler policies $\Sigma$ to keep $\rho$ within the indistinguishability budget required by Theorem~\ref{thm:composition}. Detailed operational trade-offs and stress-test results appear in Section~\ref{sec:security}. We now formalize the adversary’s view and prove a composition bound that links the scheduler’s distortion and per-channel footprints to the warden’s advantage (Section~\ref{sec:security}).

\subsection{Security Assumptions}
\label{Assumptions}

The following assumptions underlie the formal games in
Section~\ref{subsec:formal}, the composition result in
Theorem~\ref{thm:composition}, and the confidentiality and
integrity results in
Section~\ref{sec:adversary-operational}. Detailed formulations,
parameter bounds, and modelling limitations are provided in
Supplementary
(see Appendix~D of the Supplementary Material for details). Throughout,
$\lambda$ is the security parameter and $\mathcal{A}$ is a
probabilistic polynomial-time adversary.

\begin{enumerate}[
    label=\textbf{A\arabic*.},
    leftmargin=2.3em,
    itemsep=0.35em,
    topsep=0.35em
]

\item \textbf{Key secrecy and hash idealisation.}
The master secret and its domain-separated session keys remain
unknown to $\mathcal{A}$. Where the key-recovery argument uses
the random-oracle model, the key-derivation hash is modelled as
a consistent random oracle and its hidden input remains
computationally unpredictable given the adversary's transcript
\cite{Bellare1993,katz2014introduction}. This assumption is used
in Claims~\ref{Claim1} and~\ref{Claim2}, and in
Theorem~\ref{theoremMTM}.

\item \textbf{Cover-parameter opacity and masking
pseudorandomness.}
The synthesised objects on $C_1$ and $C_2$ satisfy
\textsf{CP-IND} and \textsf{CP-UNL}, while the keyed
cover-conditioned masking operation produces a payload whose
distinguishability from a uniform string is bounded by
$\varepsilon_{\mathsf{param}}$. The public cover tag is reproducible by both endpoints and is included in the adversary's view
\cite{@hopper2009provably,vonahn2003publickey,Goldreich2001,chen2001quantization}. This assumption is used in
Definitions~\ref{def:cp-ind} and~\ref{def:cp_unl},
Lemmas~\ref{lem:cpind} and~\ref{lem:cpunl}, and
Theorem~\ref{thm:composition}.

\item \textbf{Embedding correctness and
indistinguishability.}
For admissible covers, decoding inverts embedding except with
the stated decoding-failure probability. The distribution of the
stego-object on $C_3$ is computationally indistinguishable from
the corresponding cover distribution up to
$\varepsilon_{\mathsf{emb}}$ under the oracle access permitted
by Definition~\ref{def:ind-stego-mc}
\cite{@hopper2009provably,vonahn2003publickey}.
The cover-parameter streams on $C_1$ and $C_2$ are governed by
\textbf{A2} and are not counted again in the embedding term.
This assumption is used in Theorem~\ref{thm:composition} and
Claims~\ref{Claim2} and~\ref{Claim4}.

\item \textbf{Message-authentication security.}
HMAC is instantiated with a secret, domain-separated key and a
canonical encoding that binds the channel identifier, epoch
header, freshness value, and authenticated content. It is assumed
to be unforgeable under chosen-message attack
\cite{bellare1996keying,rfc2104}. This assumption is used in
Claim~\ref{Claim3}, Claim~\ref{Claim4}, and Theorem~\ref{theoremMTM}. Replay rejection additionally requires the freshness assumption in \textbf{A5}.

\item \textbf{Nonce freshness and receiver state.}
The per-channel freshness values are generated using an approved
random-bit-generation mechanism, are not reused under the same
authentication key, and are checked against receiver-maintained
epoch or nonce state \cite{nist80090ar1}. This assumption is used
by Algorithm~\ref{alg:mc-sync}, and Claim~\ref{Claim3}. A verbatim replay is rejected by the
freshness check; random nonce collisions are treated separately
in the detailed supplementary analysis.

\item \textbf{Admissible scheduling.}
The scheduler belongs to the admissible class $\Sigma$ and
satisfies the delay, skew, burst, timeout, and metadata conditions
required by Algorithm~\ref{alg:mc-sync}. Its residual
cover-versus-protocol timing footprint is bounded by
$\varepsilon_{\mathsf{sched}}(\sigma^\star)$
\cite{levin2017markov}. This assumption is used in
Lemma~\ref{lem:resync}, Lemma~\ref{lem:cpunl}, and
Theorem~\ref{thm:composition}.

\end{enumerate}

The assumptions exclude endpoint compromise, implementation
side channels, entropy-source failure, and denial-of-service
conditions that prevent protocol delivery altogether.

\section{Security Properties and Discussion}
\label{sec:security}

This section consolidates the guarantees provided by our hybrid protocol $\mathcal{P}^{\mathsf{cs,\,cm}}_{\mathsf{hyb-stego}}$ against multichannel adversaries. We analyse security under the assumptions in Section~\ref{Assumptions}, using the multichannel indistinguishability game and oracles formalised in \S\ref{subsec:formal}, and the hybrid composition bound of Theorem~\ref{thm:composition} in \S\ref{subsec:main-theorem}. Our analysis follows the provable-security tradition in steganography (e.g., the Hopper–Langford–von~Ahn framework and its chosen-hiddentext variants) adapted to multichannel traces and timing/flow shaping \cite{@hopper2009provably,@Hopper2002, vonahn2003publickey}. 

Concretely, we prove: (i) \emph{cover-parameter innocuity} (CP-IND) and Cover-Parameter Unlinkability (CP-UNL), showing that the parameters $(\gamma_1,\gamma_2)$ are computationally indistinguishable from native cover artifacts and reveal neither the presence of steganography nor helpful information about $m$; (ii) \emph{multichannel indistinguishability} (IND-STEGO-MC), bounding a PPT warden’s advantage by the sum of three design terms: $\varepsilon_{\mathsf{param}}$ (masking/synthesis), $\varepsilon_{\mathsf{emb}}$ (per-channel embedding), and $\varepsilon_{\mathsf{sched}}$ (scheduler shaping) plus negligible slack; and (iii) \emph{robustness/integrity}, showing that multichannel replay and MitM modifications are rejected with negligible probability by nonce freshness and (S)UF-CMA MAC security \cite{rfc2104,bellare1996keying}. 

\subsection{Formal Framework and Security Definition}
\label{subsec:formal}

\paragraph{Problem statement.}
A sender $\mathcal{U}$ transmits a secret $m\!\in\!\{0,1\}^*$ to a receiver $\mathcal{R}$ across $c\!\ge\!2$ channels $\mathcal{C}=\{1,\dots,c\}$ while an adaptive PPT warden $\mathcal{A}$ observes multi-channel traces and may shape timing within admissible bounds. Let $\Sigma$ be the class of \emph{admissible schedulers} (bounded delay/jitter/burst); for $\sigma\!\in\!\Sigma$, the (possibly randomized) mapper $\mathsf{Sched}(\cdot;\sigma)$ emits per-channel objects as timed traces.

A hybrid stegosystem $\mathcal{S}_{\mathsf{CPM}}$ samples $k\!\leftarrow\!\mathsf{Setup}(1^\lambda)$, derives a cover parameter $P_{\mathsf{params}}\!\leftarrow\!\mathsf{Synth}(k,\ell)$ for $\ell\!=\!|m|$, masks $m$ to $b\!=\!F_{\mathsf{mask}}(m,P_{\mathsf{params}})$, and embeds $b$ into covers $o_i\!\sim\!\mathcal{D}_i$ via $s_i=\mathsf{Enc}(k,o_i,b)$, yielding $s=(s_1,\dots,s_c)$ that $\mathsf{Sched}$ ships as traces.

\textit{Notation bridge to the protocol (Sec.~\ref{Hybrid_Entropy_Steganographic_Communication_Protocol}).}
In the concrete protocol, the parameter realisation is conveyed as two benign strings $(\gamma_1,\gamma_2)$ over $C_1,C_2$ and the stego $s$ over $C_3$. For analysis we treat $(\gamma_1,\gamma_2)$ as the realization (split or joint) of $P_{\mathsf{params}}$. With QIM-style cover-conditioned masking (see Appendix~B of the Supplementary Material), $F_{\mathsf{mask}}$ may incorporate a quantizer with a dither derived from $(o,\gamma_1,\gamma_2,k)$; the abstraction here remains $b=F_{\mathsf{mask}}(m,P_{\mathsf{params}})$.

\paragraph{Security goals}
\begin{enumerate}[leftmargin=1.4em,label=\textbf{G\arabic*.}]
  \item \textbf{Param–cover indistinguishability (PCI).} The pair $(\gamma_1,\gamma_2)$ is computationally indistinguishable from benign channel content and independent of $m$ and of the use of $C_3$ (CP-IND/CP-UNL).
  \item \textbf{Multi-channel indistinguishability (IND–STEGO–MC).} For any $\sigma\!\in\!\Sigma$, $\mathcal{A}$ cannot distinguish scheduled stego traces $\mathsf{Sched}([c],\mathbf{s};\sigma)$ from scheduled genuine-cover traces beyond negligible advantage.
  \item \textbf{Correctness/robustness.} $\mathsf{Dec}(k,\cdot)$ and $F_{\mathsf{unmask}}$ recover $m$ with high probability under admissible loss/skew for any $\sigma\!\in\!\Sigma$ (cf. Secs.~\ref{proto_Des} and \ref{subsec:algo-sync}).
\end{enumerate}

\paragraph{Game view}
Fix distributions $\mathcal{D}_i$ for covers on each channel. The IND–STEGO–MC game grants $\mathcal{A}$ oracle access to either stego or cover traces scheduled under $\sigma\!\in\!\Sigma$; $\mathcal{A}$ outputs a bit guessing the oracle type. The advantage is
\begin{multline*}
\mathsf{Adv}^{\mathsf{mc}}_{\mathcal{S}_{\mathsf{CPM}},\mathbf{\mathcal{D}}}(\mathcal{A};\lambda)
= \Bigl|\Pr[\mathcal{A}^{\mathsf{Sched}([c],\mathbf{s};\sigma)}\!\Rightarrow\!1]
 \\- \Pr[\mathcal{A}^{\mathsf{Sched}([c],\mathbf{o};\sigma)}\!\Rightarrow\!1]\Bigr|,  
\end{multline*}
where $\mathbf{o}\!\sim\!\mathcal{D}_1\times\cdots\times\mathcal{D}_c$ and $\mathbf{s}$ is generated as above. Security requires this advantage to be negligible for all PPT $\mathcal{A}$ and admissible $\sigma$.

\begin{definition}[IND-STEGO-MC]
\label{def:ind-stego-mc}
Fix $\lambda$. Let $\mathbf{\mathcal{D}}=(\mathcal{D}_1,\ldots,\mathcal{D}_c)$ be per-channel cover laws over $o_i\in\mathcal{O}_i$. The challenger samples $k\leftarrow\mathsf{Setup}(1^\lambda)$. $\mathcal{A}$ has oracle access:

\begin{itemize}
\item \emph{Cover oracle} $\mathcal{O}_{\mathsf{cov}}(i)$: draw $o_i\leftarrow\mathcal{D}_i$ and return $\mathsf{Sched}(i,o_i;\sigma)$ for current $\sigma\in\Sigma$.
\item \emph{Embedding oracle} $\mathcal{O}_{\mathsf{emb}}(m)$ (freshness enforced):
compute $P_{\mathsf{params}}\!\leftarrow\!\mathsf{Synth}(k,|m|)$,
set $b\!=\!F_{\mathsf{mask}}(m,P_{\mathsf{params}})$,
draw $\mathbf{o}\!\leftarrow\!\prod_i\mathcal{D}_i$,
derive per–channel dither $d_i\!\leftarrow\!\mathsf{Dither}(k, o_i)$,
and return $\mathsf{Sched}([c],\mathbf{s};\sigma)$ with
$s_i\!=\!\mathsf{Enc}(k,o_i,b; d_i)$.

\item \emph{Scheduler oracle} $\mathcal{O}_{\mathsf{sch}}(\sigma')$: update $\sigma\!\in\!\Sigma$ (admissible shaping: bounded delay/burst/jitter).
\end{itemize}

\noindent\textbf{Challenge.}
$\mathcal{A}$ submits $(m_0,m_1)$ with $|m_0|=|m_1|$ and $\sigma^\star\in\Sigma$. Challenger picks $\zeta \leftarrow\{0,1\}$ and $\beta\leftarrow\{0,1\}$. If $\beta=1$ (stego): compute $P_{\mathsf{params}}$, 
$b^\star\!=\!F_{\mathsf{mask}}(m_b,P_{\mathsf{params}})$,
draw $\mathbf{o}\!\leftarrow\!\prod_i\mathcal{D}_i$,
derive $d_i\!\leftarrow\!\mathsf{Dither}(k,o_i)$,
set $s_i\!=\!\mathsf{Enc}(k,o_i,b^\star; d_i)$,
and return $\mathsf{Sched}([c],\mathbf{s};\sigma^\star)$. If $\beta=0$ (cover): draw $\mathbf{o}'\!\leftarrow\!\prod_i\mathcal{D}_i$ with matched metadata and return $\mathsf{Sched}([c],\mathbf{o}';\sigma^\star)$. $\mathcal{A}$ continues queries (freshness respected) and outputs $(\hat{\beta},\hat{\zeta})$.

\noindent\textbf{Advantage.}
Let
\[
\mathsf{Adv}^{\mathsf{mc}}_{\mathcal{S}_{\mathsf{CPM}},\mathbf{\mathcal{D}}}(\mathcal{A};\lambda)
=\Big|\Pr[\hat{\beta}=\beta]-\tfrac{1}{2}\Big|.
\]
$\mathcal{S}_{\mathsf{CPM}}$ is IND–STEGO–MC secure if the first is negligible for all PPT $\mathcal{A}$; confidentiality improves as the second approaches zero.
\end{definition}

\begin{definition}[CP-IND \& CP-UNL (cover-parameter opacity)]
\label{def:cp-opacity}
Fix $\lambda$, distributions $\mathcal{D}^{\text{param}}_1,\mathcal{D}^{\text{param}}_2$ for benign
short texts on $C_1,C_2$, and a scheduler policy set $\Sigma$. We require that the per–epoch headers $H_e$ and any C$_3$ dither seeds 
$d_i\!\leftarrow\!\mathsf{Dither}(k,o_i)$ depend only on $k$ and the local cover $o_i$ and are independent of $(\gamma_1,\gamma_2)$ beyond metadata matching, so that parameter streams do not leak linkage to the payload stream.

\emph{CP-IND (indistinguishability).}
Let $\mathcal{G}_\lambda$ be the real generator that outputs
$(\gamma_1,\gamma_2)\leftarrow\mathsf{Synth}(k,\ell)$ (keyed) and ships
$H_e\|(\mathsf{nonce}_a\|\gamma_1)$ on $C_1$ and
$H_e\|(\mathsf{nonce}_b\|\gamma_2)$ on $C_2$ using $\sigma\in\Sigma$.
Let $\mathcal{H}_\lambda$ be the ideal generator that draws
$(\tilde\gamma_1,\tilde\gamma_2)\leftarrow
\mathcal{D}^{\text{param}}_1\times\mathcal{D}^{\text{param}}_2$
and ships them with matched headers/metadata under the same $\sigma$.
An adversary/warden $\mathcal{A}$ has \emph{CP-IND advantage}
\[
\mathsf{Adv}^{\mathsf{CP\text{-}IND}}(\mathcal{A})
=\Big|\Pr[\mathcal{A}(\mathcal{G}_\lambda)\!=\!1]-\Pr[\mathcal{A}(\mathcal{H}_\lambda)\!=\!1]\Big|.
\]
We say CP-IND holds if this advantage is negligible in $\lambda$.

\emph{CP-UNL (unlinkability).}
Consider the following game. Challenger samples a bit $\beta$.
If $\beta{=}1$ (linked), it produces one consistent epoch $e$ with
$(\gamma_1,\gamma_2)$ as above and a payload stego segment $s_e$
constructed from the same session (same $k$, same epoch header $H_e$),
then ships all three under the chosen $\sigma\in\Sigma$.
If $\beta{=}0$ (independent), it draws $(\tilde\gamma_1,\tilde\gamma_2)$
from $\mathcal{D}^{\text{param}}_{1,2}$, and $s'_e$ from a cover/stego
trace generated independently with matched metadata and the same $\sigma$.
The adversary outputs $\hat\beta$. Its \emph{CP-UNL advantage} is
$\big|\Pr[\hat\beta{=}\beta]-\tfrac12\big|$. CP-UNL holds if this is negligible.
\end{definition}

\subsection{Main Theorem (Composition Bound)}
\label{subsec:main-theorem}

For channel $i$, let $\mathsf{Law}(S_i)$ be the distribution of the stego object $s_i=\mathsf{Enc}(k,o_i,b)$ when $o_i\sim\mathcal{D}_i$ and $b=F_{\mathsf{mask}}(m,P_{\mathsf{params}})$ with $P_{\mathsf{params}}\leftarrow\mathsf{Synth}(k,|m|)$. Define
\begin{multline*}
    \varepsilon_{\mathsf{param}}\triangleq \mathbb{E}\Big[\mathrm{TV}\!\big(F_{\mathsf{mask}}(m,P_{\mathsf{params}}),U_\ell\big)\Big],\\
\varepsilon_{\mathsf{emb}}\triangleq \sum_{i=1}^c \mathrm{TV}\!\big(\mathsf{Law}(S_i),\mathcal{D}_i\big),
\end{multline*}
where $\mathbb{E}\big[ \cdot\big]$ is the expectation (average value), $ \varepsilon_{\mathsf{param}}$ and $ \varepsilon_{\mathsf{emb}}$ are the average statistical leakage caused by masking and embedding, respectively. $\mathrm{TV}(\cdot, \cdot)$ is the statistical distances. For admissible $\sigma\in\Sigma$, let $\varepsilon_{\mathsf{sched}}(\sigma)$ bound the scheduler’s distributional distortion on multi-channel traces (cover vs.\ stego) under matched metadata.

\paragraph{Cover-parameter security notions (for $C_1,C_2$).}
Let $\mathcal{D}^{\text{param}}_{1,2}$ denote the (benign) per-channel laws for cover texts that $C_1$ and $C_2$ carry in the ambient system (e.g., innocuous snippets), and let $\mathsf{Synth}(k,\ell)$ be the keyed sampler that generates the protocol’s cover parameters. We formalize two properties:

\begin{definition}[Cover-Parameter Indistinguishability, \textsf{CP-IND}]
\label{def:cp-ind}
Consider the following experiment. The challenger samples $b\leftarrow\{0,1\}$, a session key $k\leftarrow\mathsf{Setup}(1^\lambda)$, and a length $\ell$.
If $b=1$ (\emph{param} world), it returns $(\gamma_1,\gamma_2)\leftarrow \mathsf{Synth}(k,\ell)$; 
if $b=0$ (\emph{cover} world), it returns $(\gamma_1,\gamma_2)\leftarrow \mathcal{D}^{\text{param}}_{1,2}$ independently of any message. 
An adversary $\mathcal{A}$ outputs $\hat{b}$. Its advantage is
\[
\mathsf{Adv}^{\textsf{CP-IND}}(\mathcal{A}) \;=\; \Big|\,\Pr[\hat{b}=b]-\tfrac12\,\Big|.
\]
We say \textsf{CP-IND} holds if this advantage is negligible in $\lambda$.
\end{definition}

\begin{definition}[Cover-Parameter Unlinkability, \textsf{CP-UNL}]
\label{def:cp_unl}
Fix an epoch $e$ and header $H_e=(e,N_e)$ as in Alg.~\ref{alg:mc-sync}. In the \emph{linked} world, the challenger forms a valid stego epoch with $(\gamma_1,\gamma_2)\leftarrow\mathsf{Synth}(k,\ell)$ on $C_1,C_2$ and $(\mathsf{nonce}_c,s_e,\mathrm{MAC}_e)$ on $C_3$ (per the protocol). In the \emph{unlinked} world, $(\widetilde{\gamma}_1,\widetilde{\gamma}_2)$ are drawn from $\mathcal{D}^{\text{param}}_{1,2}$ and are independent of the $(\cdot)$ on $C_3$; the headers $H_e$ remain syntactically valid but independent of content. Given the three per-channel frames (with matched metadata), $\mathcal{A}$ outputs a bit distinguishing the worlds. The \textsf{CP-UNL} advantage is defined analogously to \textsf{CP-IND}, and \textsf{CP-UNL} holds if the advantage is negligible.
\end{definition}

\noindent These notions remain unchanged when $F_{\mathsf{mask}}$ uses cover-conditioned QIM dither: the parameter sampler $\mathsf{Synth}$ governs the distributional opacity on $C_1,C_2$, while masking pseudorandomness is captured by $\varepsilon_{\mathsf{param}}$ in the main bound.

\noindent The notions above capture the paper’s intent: (1) the items on $C_1,C_2$ must look innocuous on their own (\textsf{CP-IND}) and (2) must not betray that they are inputs to the stego on $C_3$ (\textsf{CP-UNL}). These mirror standard indistinguishability ideals for provable steganography \cite{@Hopper2002} and are compatible with PRF/PRNG-based samplers and MAC-based integrity \cite{Goldreich2001}.

\begin{lemma}[Realizing \textsf{CP-IND}]
\label{lem:cpind}
If $\mathsf{Synth}(k,\ell)$ is a domain-adapted PRF/PRNG sampler whose output distribution is computationally indistinguishable from $\mathcal{D}^{\text{param}}_{1,2}$, then for all PPT $\mathcal{A}$,
\[
\mathsf{Adv}^{\textsf{CP-IND}}(\mathcal{A}) \;\le\; \varepsilon_{\mathsf{param}},
\]
where $\varepsilon_{\mathsf{param}} \triangleq \mathbb{E}\!\left[\mathrm{TV}\!\big(F_{\mathsf{mask}}(m,P_{\mathsf{params}}),U_\ell\big)\right]$ is the (already defined) masking pseudo\-randomness term. \emph{Sketch.} Replace $F_{\mathsf{mask}}(m,P_{\mathsf{params}})$ by $U_\ell$ (cost $\varepsilon_{\mathsf{param}}$), and couple the sampler’s output to $\mathcal{D}^{\text{param}}_{1,2}$; PRF/PRNG security and distributional matching bound the distinguishing gap.

\end{lemma}
\begin{lemma}[Realizing \textsf{CP-UNL}]
\label{lem:cpunl}
If epoch headers $H_e$ are independent of content \emph{(no content fields or lengths derived from $\gamma_{1},\gamma_{2},s_e$)}, and the only cross-channel binding is the (unforgeable) MAC on $C_3$ (\S\ref{Hybrid_Entropy_Steganographic_Communication_Protocol}), then for all PPT $\mathcal{A}$,
\[
\mathsf{Adv}^{\textsf{CP-UNL}}(\mathcal{A}) \;\le\; \varepsilon_{\mathsf{sched}}(\sigma^\star) + \negl(\lambda),
\]
where $\varepsilon_{\mathsf{sched}}(\sigma^\star)$ is the admissible scheduling distortion term.
\emph{Sketch.} With content-independent $H_e$ \emph{(no content fields or lengths derived from $\gamma_{1},\gamma_{2},s_e$)} and no features shared across channels except the (unforgeable) MAC on $C_3$, any linkage test reduces to distinguishing scheduled traces under the same $\sigma^\star$, which is bounded by $\varepsilon_{\mathsf{sched}}(\sigma^\star)$; MAC unforgeability prevents constructing correlated forgeries.
\end{lemma}

\noindent Define the aggregate cover-parameter term
\[
\varepsilon_{\mathsf{cp}} \;\triangleq\; \sup_{\mathcal{A}}\Big(\mathsf{Adv}^{\textsf{CP-IND}}(\mathcal{A}) + \mathsf{Adv}^{\textsf{CP-UNL}}(\mathcal{A})\Big),
\]
so that small \textsf{CP-IND}/\textsf{CP-UNL} advantages fold into the main bound.

\begin{theorem}[Hybrid composition]
\label{thm:composition}
For any PPT $\mathcal{A}$ and any $\sigma^\star\in\Sigma$,
\begin{align*}
  \mathsf{Adv}^{\mathsf{mc}}_{\mathcal{S}_{\mathsf{CPM}},\mathbf{\mathcal{D}}}(\mathcal{A};\lambda)
\le
\underbrace{\varepsilon_{\mathsf{param}}+\varepsilon_{\mathsf{emb}}+\varepsilon_{\mathsf{sched}}(\sigma^\star)}_{\text{as in \S\ref{subsec:main-theorem}}}
+\underbrace{\varepsilon_{\mathsf{cp}}}_{\textsf{CP-IND}+\textsf{CP-UNL}}
\\+\;
\negl(\lambda).  
\end{align*}

If each term is negligible in $\lambda$ (or $\le\delta_i$ by design), then $\mathcal{S}_{\mathsf{CPM}}$ is IND-STEGO-MC secure.
\end{theorem}

\paragraph{Refinement (cover-parameter opacity).}
Let
\begin{multline*}
    \varepsilon_{\mathsf{cp\text{-}ind}}
\;\triangleq\;
\mathrm{TV}\!\big(\mathsf{Law}(\gamma_1,\gamma_2),\;
\mathcal{D}^{\text{param}}_1\times\mathcal{D}^{\text{param}}_2\big),
\\
\varepsilon_{\mathsf{cp\text{-}unl}}
\;\triangleq\;
\sup_{\sigma\in\Sigma}\;\mathsf{Adv}^{\mathsf{CP\text{-}UNL}}_\sigma.
\end{multline*}
Then the multi-channel advantage satisfies
\begin{align*}
    \mathsf{Adv}^{\mathsf{mc}}_{\mathcal{S}_{\mathsf{CPM}},\mathbf{\mathcal{D}}}(\mathcal{A};\lambda)
\;&\le\;
\underbrace{\varepsilon_{\mathsf{cp\text{-}ind}}}_{\text{param.\ innocuous}}
+\underbrace{\varepsilon_{\mathsf{cp\text{-}unl}}}_{\text{param.\ unlinkability}}
\\ &+\underbrace{\varepsilon_{\mathsf{param}}}_{\text{mask pseudorandomness}}
+\underbrace{\varepsilon_{\mathsf{emb}}}_{\text{embedding footprint}}
\\&+\underbrace{\varepsilon_{\mathsf{sched}}(\sigma^\star)}_{\text{admissible shaping}}
+\negl(\lambda).
\end{align*}

\paragraph{Proof sketch.}
Hybrid chain: (H0) real stego → (H1) replace $b$ by $U_\ell$ (cost $\varepsilon_{\mathsf{param}}$ since $F_{\mathsf{mask}}(m,P_{\mathsf{params}})$ is pseudorandom if $P_{\mathsf{params}}$ is) → (H2) replace each $s_i$ by a draw from $\mathcal{D}_i$ coupled to $\mathsf{Law}(S_i)$ (cost $\varepsilon_{\mathsf{emb}}$ via TV/coupling) → (H3) align scheduled stego traces to scheduled cover traces under the same $\sigma^\star$ (cost $\varepsilon_{\mathsf{sched}}(\sigma^\star)$ by admissibility) → (H4) cover world. Triangle inequality yields the bound; $\negl(\lambda)$ accounts for PRF/PRNG security and oracle freshness.
\paragraph{Corollaries.}
\textbf{(i) Chosen-covertext / chosen-flow.} Oracle access $\{\mathcal{O}_{\mathsf{cov}},\mathcal{O}_{\mathsf{emb}},\mathcal{O}_{\mathsf{sch}}\}$ is simulatable in the hybrids with the same costs, so Theorem~\ref{thm:composition} holds unchanged.  
\textbf{(ii) Partial-channel compromise.} If $\mathcal{A}$ compromises $T\subset[c]$, then the bound applies over $[c]\setminus T$:
\[
\mathsf{Adv}^{\mathsf{mc}}\!\le\!\sum_{i\in[c]\setminus T}\!\mathrm{TV}(\mathsf{Law}(S_i),\mathcal{D}_i)+\varepsilon_{\mathsf{param}}+\varepsilon_{\mathsf{sched}}(\sigma^\star)+\negl(\lambda).
\]
\textbf{(iii) Budgeted shaping.} If $\Sigma$ enforces per-trace distortion budget $\rho$, then $\varepsilon_{\mathsf{sched}}(\sigma)\le\rho$ for all $\sigma\in\Sigma$, giving
\(
\mathsf{Adv}^{\mathsf{mc}}\le \varepsilon_{\mathsf{param}}+\varepsilon_{\mathsf{emb}}+\rho+\negl(\lambda).
\)

\paragraph{Pointer to practice.}
Operational implications (e.g., picking ECC, window/timeout parameters, heartbeat cadence, and admissible jitter) are evaluated empirically and mapped to the bound terms in Section~\ref{sec:security}.

\subsection{Adversary Model and Operational Threat Scenarios}
\label{sec:adversary-operational}

The \textsf{IND-STEGO-MC} game in
Section~\ref{subsec:formal} evaluates whether a warden can
distinguish protocol traces from benign multichannel traffic. The active analysis below instead considers message recovery and the
acceptance of adversarially modified transcripts. These objectives
are analysed separately because detectability does not necessarily
imply message disclosure, while statistical undetectability alone
does not guarantee integrity. Complete challenger algorithms,
receiver predicates, reductions, and query-dependent bounds are
provided in Supplementary
Appendix~E of the Supplementary Material.
\begin{definition}[Active Multichannel Attack]
\label{definition7}
An \emph{active multichannel attack}
($\mathsf{MC\mbox{-}ATTACK}$) is an interaction between a PPT adversary $\mathcal{A}$ and
$\mathcal{P}^{\mathsf{cs,cm}}_{\mathsf{hyb\text{-}stego}}$ across $C_1,\ldots,C_c$, where $\mathcal{A}$ observes all channel contents and public metadata and may delay, reorder, replay, modify, splice, drop, or inject frames. The adversary succeeds by recovering the concealed message beyond its guessing baseline or
causing the receiver to accept an unauthorised transcript. Denial-of-service through frame suppression is excluded outside the claimed security objectives. 
\end{definition}

\begin{game}[Active Multichannel Attack Game]
\label{game1}
The experiment between the challenger $\mathsf{Chal}$, the
adversary $\mathcal{A}$, and the receiver proceeds as follows.

\begin{enumerate}[label=\Alph*.,leftmargin=2.2em]
\item \textbf{Initialisation.}
$\mathsf{Chal}$ runs $\mathsf{Setup}(1^\lambda)$, derives the required protocol keys, initialises receiver freshness state, and establishes $C_1,C_2,C_3$ with scheduler $\sigma\in\Sigma$.
\item \textbf{Session generation.}
For each session or epoch $e$, $\mathsf{Chal}$ samples a message
$m_e$ from the experiment's message distribution, generates
$(\gamma_{1,e},\gamma_{2,e})$, masks and embeds the message, and
obtains the stego-object $s_e$. It constructs
\[
H_e=(e,N_e),
\qquad
N_e=\mathsf{HMAC}_{K_{\mathsf{hdr}}}
(\texttt{hdr}\parallel e),
\]
and forms the three channel frames
\begin{align}\label{complete_protocol_frame}
\mathcal{F}_{1,e}
&=
H_e\parallel
(\mathsf{nonce}_{a,e}\parallel\gamma_{1,e}),
\nonumber\\
\mathcal{F}_{2,e}
&=
H_e\parallel
(\mathsf{nonce}_{b,e}\parallel\gamma_{2,e}),
\nonumber\\
\mathcal{F}_{3,e}
&=
H_e\parallel
(\mathsf{nonce}_{c,e}\parallel s_e
\parallel\mathsf{MAC}_e)
\end{align}
The authentication tag binds the canonical epoch transcript,
including the header, the three freshness values, both
cover-parameter objects, and the stego-object. The scheduled
frames and their public metadata are exposed to $\mathcal{A}$.

\item \textbf{Adversary interaction.}
$\mathcal{A}$ may observe many valid sessions and adaptively may replay, modify, or inject frames on any channel, combining components from different epochs or sessions. Timing manipulations are subject to
$\Sigma$.

\item \textbf{Receiver processing.}
The receiver applies Algorithm~\ref{alg:mc-sync} to align the submitted frames. It verifies the epoch header, nonce freshness, canonical encoding, and authentication of the complete epoch transcript. If any required component is missing, stale, malformed, or unauthenticated, it outputs $\bot$; otherwise, it decodes and unmasks the stego-object and outputs $m'_e$.

\item \textbf{Outcome.}
In the message-recovery experiment, $\mathcal{A}$ outputs a guess $\widehat m_e$. In the substitution experiment, the adversary wins if the receiver accepts an epoch for which
$m'_e\neq m_e$. 

\end{enumerate}
\end{game}

\paragraph{Security objectives.}
Let $\mathsf{View}_{\mathcal{A}}$ contain the adversary's complete
multichannel observation, public metadata, oracle responses, and
auxiliary information. Its prior-guessing probability is
\begin{equation}
\label{eq:message_guessing_baseline}
p_{\mathsf{base}}
=
2^{-H_{\infty}
(M\mid\mathsf{View}_{\mathcal{A}})}.
\end{equation}
For the confidentiality objective, the advantage is
\begin{equation}
\label{eq:message_recovery_advantage}
\underbrace{\mathsf{Adv}^{\mathsf{MR}}_{\mathcal{A}}(\lambda)}_{\text{message recovery}}
=
\max\!
\left\{
0,\,
\Pr[\widehat m=m]-p_{\mathsf{base}}
\right\}.
\end{equation}
For a uniformly distributed message independent of the
adversary's prior view, $p_{\mathsf{base}}=2^{-|m|}$.

The authenticated-substitution (Integrity objective) advantage is
\begin{equation}
\label{eq:message_substitution_advantage}
\underbrace{\mathsf{Adv}^{\mathsf{SUB}}_{\mathcal{A}}(\lambda)}_{\text{authenticated-substitution}}
=
\Pr\!\left[
\exists e:
\mathsf{Accept}_e=1
\wedge
m'_e\neq m_e
\right].
\end{equation}
This event includes modification of the stego-object or either
cover-parameter object, cross-channel replacement, and
cross-epoch or cross-session splicing. Section~\ref{Confidentiality Analysis}
bounds
$\operatorname{Adv}^{\mathsf{MR}}_{\mathcal{A}}$, while
Section~\ref{sec:integrity} bounds
$\operatorname{Adv}^{\mathsf{SUB}}_{\mathcal{A}}$ and replay
acceptance.

\subsection{Confidentiality Analysis}
\label{Confidentiality Analysis}

We analyse confidentiality through the message-recovery
experiment of Game~\ref{game1}. This objective concerns recovery
of the concealed plaintext rather than detection of the protocol
trace. The analysis relies principally on key secrecy and the
pseudorandomness of the cover-conditioned masking operation.
Complete reductions, conditioning arguments, and
query-dependent derivations are provided in Supplementary
Appendix~E of the Supplementary Material.

\begin{claim}[Hardness of Stego-Key Recovery]
\label{Claim1}
Under the hidden-input unpredictability and random-oracle
conditions of Assumption~\textbf{A1}, any PPT adversary
$\mathcal{A}$ making at most $q_H(\lambda)$ hash-oracle queries
satisfies
\begin{equation}
\label{eq:key_recovery_bound}
\operatorname{Adv}^{\mathsf{Key}}_{\mathcal{A}}(\lambda)
\coloneqq
\Pr\!\left[
\mathcal{A}\text{ obtains }k_{\mathsf{stego}}
\right]
\leq
\frac{q_H(\lambda)}{2^{\kappa}}
+
\operatorname{negl}(\lambda),
\end{equation}
where $\kappa$ is the effective unpredictability of the hidden
key-derivation input.
\end{claim}

\noindent
The bound follows because recovering $k_{\mathsf{stego}}$
requires querying the random oracle at the hidden derivation
input. Each of the adversary's $q_H(\lambda)$ queries succeeds
with probability at most $2^{-\kappa}$, and the union bound gives
\eqref{eq:key_recovery_bound}. The complete derivation is
reported in Supplementary
Appendix~E of the Supplementary Material.

\begin{claim}[Message Confidentiality under
$\mathsf{MC\mbox{-}ATTACK}$]
\label{Claim2}
Let $\mathsf{View}_{\mathcal{A}}$ denote the auxiliary
information available to $\mathcal{A}$ before generation of the
challenge transcript. Under Assumptions~\textbf{A1} and
\textbf{A2}, the probability that any PPT adversary recovers the
challenge message satisfies
\begin{align}
\Pr[\widehat M=M]
\leq\;&
2^{-H_{\infty}
\left(
M\mid\mathsf{View}_{\mathcal{A}}
\right)}
+
\operatorname{Adv}^{\mathsf{Key}}_{\mathcal{A}}(\lambda)
\nonumber\\
&+
\varepsilon_{\mathsf{param}}
+
\operatorname{negl}(\lambda).
\label{eq:message_confidentiality_bound}
\end{align}
Equivalently, relative to the prior-guessing baseline defined in
\eqref{eq:message_guessing_baseline},
\begin{equation}
\label{eq:message_recovery_bound}
\operatorname{Adv}^{\mathsf{MR}}_{\mathcal{A}}(\lambda)
\leq
\operatorname{Adv}^{\mathsf{Key}}_{\mathcal{A}}(\lambda)
+
\varepsilon_{\mathsf{param}}
+
\operatorname{negl}(\lambda).
\end{equation}
Combining Claims~\ref{Claim1} and~\ref{Claim2} therefore gives
\begin{align}
\operatorname{Adv}^{\mathsf{MR}}_{\mathcal{A}}(\lambda)
\leq
\frac{q_H(\lambda)}{2^{\kappa}}
+
\varepsilon_{\mathsf{param}}
+
\operatorname{negl}(\lambda).
\label{eq:combined_message_recovery_bound}
\end{align}
\end{claim}

\begin{proof}[Proof sketch]
Consider a hybrid experiment in which the keyed
cover-conditioned mask used for the challenge message is replaced
by an independent uniform string. By Assumption~\textbf{A2}, this
replacement changes the adversary's success probability by at
most
$\varepsilon_{\mathsf{param}}+\operatorname{negl}(\lambda)$.
In the resulting hybrid, the masked payload is independent of
$M$. Conditioned on the adversary not recovering
$k_{\mathsf{stego}}$, its optimal strategy is therefore prior
guessing, with success probability
\[
2^{-H_{\infty}
\left(
M\mid\mathsf{View}_{\mathcal{A}}
\right)}.
\]
Adding the key-recovery event and the masking replacement cost
yields \eqref{eq:message_confidentiality_bound}; subtracting the
prior-guessing baseline gives
\eqref{eq:message_recovery_bound}.
\end{proof}

The terms
$\varepsilon_{\mathsf{cp\text{-}ind}}$,
$\varepsilon_{\mathsf{cp\text{-}unl}}$,
$\varepsilon_{\mathsf{emb}}$, and
$\varepsilon_{\mathsf{sched}}(\sigma^\star)$ quantify
detectability and cross-channel linkage in
Theorem~\ref{thm:composition}. They are not included again in
the message-recovery bound because detecting or linking a
protocol trace does not by itself reveal the concealed plaintext.

\subsection{Integrity Analysis}
\label{sec:integrity}

Having demonstrated in Section~\ref{Confidentiality Analysis} that the protocol $\mathcal{P}^{\mathsf{cs,\,cm}}_{\mathsf{hyb-stego}}$ preserves message confidentiality, we now examine its resilience against attacks aimed at \emph{violating integrity}. Specifically, we assess whether a multichannel adversary $\mathcal{A}$ can manipulate, replay, or forge transmissions without detection. Two canonical threats are considered: the \emph{multichannel replay attack} and the \emph{multichannel man-in-the-middle attack (MC-MitM)}. Both focus on undermining the \emph{authenticity} and \emph{freshness} of the communicated data.

\subsubsection{Multichannel Replay Attack}
\label{sec:multichannel_replay_attack}

In a multichannel replay attack, the adversary
$\mathcal{A}$ records one or more valid protocol frames and later
reintroduces them on the same channel, on a different channel, or
as part of a transcript assembled from several epochs or sessions.
The objective is to cause the receiver to accept a previously
processed message, or a replay-derived combination of channel
components, as a fresh protocol epoch. The analysis below
considers complete-epoch replay, individual-channel replay,
cross-epoch splicing, and attempts to make recorded content appear
fresh.

\begin{claim}[Security Against Multichannel Replay]
\label{Claim3}
Assume that:

\begin{enumerate}[label=(\roman*),leftmargin=2em]
    \item each channel uses a $\nu$-bit freshness value that is
    never intentionally reused under the same session keys;

    \item the receiver maintains authenticated epoch and nonce
    state and rejects every previously accepted epoch or freshness
    value, as required by Algorithm~\ref{alg:mc-sync} and
    Assumption~\textbf{A5};

    \item the canonical epoch transcript binds the header, the
    three freshness values, both cover-parameter objects, and the
    stego-object; and

    \item the corresponding epoch tag is EUF-CMA secure under the
    domain-separated authentication key, as required by
    Assumption~\textbf{A4}
    \cite{bellare1996keying,rfc2104}.
\end{enumerate}

Then, for any PPT adversary $\mathcal{A}$ observing at most
$Q\in\operatorname{poly}(\lambda)$ legitimate epochs and making
at most $q_v(\lambda)$ receiver-verification attempts,
\begin{align}
\operatorname{Adv}^{\mathsf{MC\text{-}REPLAY}}_{\mathcal A}
(\lambda)
\leq\;&
\operatorname{Adv}^{\mathsf{EUF\text{-}CMA}}_{\mathsf{HMAC}}
(\mathcal B)
+
\frac{3Q(Q-1)}{2^{\nu+1}}
\nonumber\\
&+
\delta_{\mathsf{state}}
+
\operatorname{negl}(\lambda),
\label{eq:mc-replay-bound}
\end{align}
where $\delta_{\mathsf{state}}$ denotes the probability of a
receiver-state failure, rollback, or loss. Hence, for polynomial
$Q$ and $q_v$, sufficiently large $\nu$, a secure MAC, and
reliable replay state, the replay-acceptance probability is
negligible in $\lambda$.
\end{claim}

\begin{proof}[Proof of Claim~\ref{Claim3}]
\textbf{Authenticated epoch construction.}
For epoch $e$, the sender transmits the complete protocol frame in \eqref{complete_protocol_frame} over $C_1,C_2,C_3$, respectively. Define the canonical epoch
transcript as
\begin{align}
\mathcal{T}_e
=
\mathsf{Encode}_{\mathsf{epoch}}\big(
&H_e,
\mathsf{nonce}_{a,e},\gamma_{1,e},
\mathsf{nonce}_{b,e},\gamma_{2,e},
\nonumber\\[-0.2em]
&\mathsf{nonce}_{c,e},s_e
\big),
\label{eq:replay-canonical-transcript}
\end{align}
and let
\begin{equation}
\mathsf{MAC}_e
=
\mathsf{HMAC}_{K_{\mathsf{mac}}}
\left(
\texttt{epoch}\parallel\mathcal{T}_e
\right).
\label{eq:replay-transcript-mac}
\end{equation}
The tag may be carried on $C_3$, but it authenticates the
contents received over all three channels. The receiver accepts
only when the frames have the same valid epoch header, all three
freshness values pass the replay-state checks, and
\eqref{eq:replay-transcript-mac} verifies.

Let $\mathsf{ReplayAccept}$ denote the event that the receiver
accepts a transcript containing at least one component copied
from an earlier accepted epoch. We partition the event into the
following cases:
\begin{equation}
\mathsf{ReplayAccept}
\subseteq
E_{\mathsf{verbatim}}
\cup
E_{\mathsf{collision}}
\cup
E_{\mathsf{forge}}
\cup
E_{\mathsf{state}}.
\label{eq:replay-event-decomposition}
\end{equation}

\textbf{Case 1: verbatim epoch replay.}
Suppose $\mathcal{A}$ resubmits all three recorded frames from an
already accepted epoch $e$. The replay contains the same epoch
identifier and the same values
$\mathsf{nonce}_{a,e}$,
$\mathsf{nonce}_{b,e}$, and
$\mathsf{nonce}_{c,e}$. Provided that the receiver retains the
corresponding epoch and nonce state, the freshness predicate
rejects the transcript deterministically. Therefore,
\begin{equation}
\Pr[
E_{\mathsf{verbatim}}
\mid
\neg E_{\mathsf{state}}
]
=0.
\label{eq:verbatim-replay-rejection}
\end{equation}
A deliberate replay is thus not accepted with probability
$2^{-\nu}$; it is rejected because it intentionally reuses
previously accepted state.

\textbf{Case 2: partial-channel replay and splicing.}
The adversary may replay only one or two frames, for example,
combining
$\mathcal{F}_{1,e_1}$,
$\mathcal{F}_{2,e_2}$, and
$\mathcal{F}_{3,e_3}$ from different epochs. If
$e_1,e_2,e_3$ are unequal, the frames fail the epoch-alignment
test. If the adversary modifies their headers so that the epoch
values agree, the reconstructed transcript differs from the
transcript authenticated by the recorded tag.

Similarly, replacing only $\gamma_{1,e}$, $\gamma_{2,e}$, or
$s_e$ changes the canonical encoding
$\mathcal{T}_e$. Consequently, acceptance of a partial replay or
spliced transcript requires a valid authentication tag for a new
canonical byte string. Such an event is included in
$E_{\mathsf{forge}}$.

\textbf{Case 3: replay made to appear fresh.}
An adversary may attempt to replace a recorded epoch identifier or
one of its nonces with a previously unused value. Any such
replacement changes an authenticated field in
\eqref{eq:replay-canonical-transcript}. The old tag therefore no
longer verifies. Producing a valid tag for the modified transcript
constitutes an existential forgery under chosen-message attack.
Hence,
\begin{equation}
\Pr[E_{\mathsf{forge}}]
\leq
\operatorname{Adv}^{\mathsf{EUF\text{-}CMA}}_{\mathsf{HMAC}}
(\mathcal B)
+
\operatorname{negl}(\lambda),
\label{eq:replay-forgery-bound}
\end{equation}
where $\mathcal B$ uses the first accepted new transcript and tag
to construct an EUF-CMA forgery.

EUF-CMA security is sufficient here because the modified epoch
transcript is a new authenticated message. Strong unforgeability
would additionally rule out producing a different valid tag for
an already authenticated transcript, but that stronger property
is not required for rejection of a verbatim replay because the
freshness check already rejects it.

\textbf{Case 4: accidental nonce collision.}
Although a deliberate replay is not a random collision, the
sender may accidentally generate a nonce that was previously
used. Suppose that each of the three channels generates at most
$Q$ independent $\nu$-bit nonces under the same session-key
lifetime. For one channel, the birthday bound gives
\begin{equation}
\Pr[\mathsf{Collision}_j]
\leq
\frac{Q(Q-1)}{2^{\nu+1}},
\qquad
j\in\{a,b,c\}.
\label{eq:per-channel-nonce-collision}
\end{equation}
Applying the union bound over the three nonce streams yields
\begin{align}
\Pr[E_{\mathsf{collision}}]
&\leq
\sum_{j\in\{a,b,c\}}
\Pr[\mathsf{Collision}_j]
\nonumber\\
&\leq
\frac{3Q(Q-1)}{2^{\nu+1}}.
\label{eq:three-channel-nonce-collision}
\end{align}
This term concerns accidental generation of repeated freshness
values; it does not describe the probability of accepting a
recorded nonce deliberately supplied by the adversary.

\textbf{Case 5: replay-state failure.}
A previously accepted transcript may be accepted again if the
receiver loses, rolls back, or incorrectly updates its persistent
epoch or nonce state. We represent this implementation-dependent
event by
\begin{equation}
\Pr[E_{\mathsf{state}}]
\leq
\delta_{\mathsf{state}}.
\label{eq:replay-state-failure}
\end{equation}
Under the formal protocol model, in which the receiver's
freshness state is reliable and updated atomically after
acceptance, $\delta_{\mathsf{state}}=0$.

\textbf{Combining the cases.}
Applying the union bound to
\eqref{eq:replay-event-decomposition} and using
\eqref{eq:verbatim-replay-rejection},
\eqref{eq:replay-forgery-bound},
\eqref{eq:three-channel-nonce-collision}, and
\eqref{eq:replay-state-failure}, we obtain
\begin{align}
\Pr[\mathsf{ReplayAccept}]
\leq\;&
\operatorname{Adv}^{\mathsf{EUF\text{-}CMA}}_{\mathsf{HMAC}}
(\mathcal B)
+
\frac{3Q(Q-1)}{2^{\nu+1}}
\nonumber\\
&+
\delta_{\mathsf{state}}
+
\operatorname{negl}(\lambda).
\end{align}
By definition,
\[
\operatorname{Adv}^{\mathsf{MC\text{-}REPLAY}}_{\mathcal A}
(\lambda)
=
\Pr[\mathsf{ReplayAccept}],
\]
which proves \eqref{eq:mc-replay-bound}. If $Q$ is polynomial,
$\nu$ grows with the security parameter, HMAC is EUF-CMA secure,
and receiver state is reliable, every term is negligible in
$\lambda$.
\end{proof}

\noindent
The result separates three mechanisms that were conflated in the
earlier proof. Freshness state rejects exact replays
deterministically; the nonce-collision term concerns only
accidental sender-side reuse; and transcript authentication
prevents an adversary from refreshing, modifying, or splicing
recorded channel components without producing a MAC forgery.

\subsubsection{Multichannel Man-in-the-Middle Attack}
\label{Multichannel Man-in-the-Middle Attack}

Following the replay analysis in
Section~\ref{sec:multichannel_replay_attack}, we consider an
active multichannel man-in-the-middle adversary that observes all
three channels and attempts to modify, replace, or splice their
components so that the receiver accepts a message
$m'_e\neq m_e$. The attack includes alteration of either
cover-parameter object, replacement of the stego-object, and
cross-channel or cross-epoch substitution.

\begin{claim}[Security Against Multichannel MitM]
\label{Claim4}
Assume that the complete canonical epoch transcript
$\mathcal{T}_e$ defined in
\eqref{eq:replay-canonical-transcript} is authenticated by
$\mathsf{MAC}_e$ as specified in
\eqref{eq:replay-transcript-mac}, that its encoding is injective
and length unambiguous, and that HMAC is EUF-CMA secure under
the domain-separated authentication key
\cite{bellare1996keying,rfc2104}. Under
Assumptions~\textbf{A4}--\textbf{A5}, any PPT adversary
$\mathcal{A}$ satisfies
\begin{align}
\operatorname{Adv}^{\mathsf{MC\text{-}MitM}}_{\mathcal A}
(\lambda)
&\coloneqq
\Pr\!\left[
\exists e:
\mathsf{Accept}_e=1
\wedge
m'_e\neq m_e
\right]
\nonumber\\
&\leq
\operatorname{Adv}^{\mathsf{EUF\text{-}CMA}}_{\mathsf{HMAC}}
(\mathcal B)
+
\frac{3Q(Q-1)}{2^{\nu+1}}
+
\delta_{\mathsf{state}} \\
&+
\operatorname{negl}(\lambda),
\label{eq:mc-mitm-bound}
\end{align}
where $Q$ is the maximum number of epochs generated during one
authentication-key lifetime, $\nu$ is the nonce length, and
$\delta_{\mathsf{state}}$ bounds failure or rollback of the
receiver's freshness state. Hence the adversary's success
probability is negligible for polynomial $Q$, sufficiently large
$\nu$, a secure MAC, and reliable receiver state.
\end{claim}

\begin{proof}[Proof sketch]
To cause acceptance of $m'_e\neq m_e$, the adversary must either
submit a new canonical transcript or reuse a previously
authenticated one. In the first case, changing
$\gamma_{1,e}$, $\gamma_{2,e}$, $s_e$, a nonce, or the epoch
header changes $\mathcal{T}_e$; acceptance therefore requires a
valid tag for a previously unauthenticated byte string and yields
an EUF-CMA forgery. In the second case, acceptance constitutes a
replay or cross-epoch splicing event and is bounded by
Claim~\ref{Claim3}, including accidental nonce collisions and
receiver-state failure. Applying the union bound gives
\eqref{eq:mc-mitm-bound}. The complete reduction and
case-by-case analysis are provided in Supplementary
Section~E-A3 of the Supplementary Material.
\end{proof}

\subsection{Security Against Multichannel Attacks}
\label{sec:CombinedMultichannelSecurity}

Claims~\ref{Claim1}--\ref{Claim4} establish the protocol's
confidentiality and integrity guarantees under the
$\mathsf{MC\mbox{-}ATTACK}$ model. 
These results are consolidated below.

\begin{theorem}[Security Against Multichannel Attacks]
\label{theoremMTM}
Let
$\mathcal{P}^{\mathsf{cs,cm}}_{\mathsf{hyb\text{-}stego}}$
be the protocol of
Section~\ref{Hybrid_Entropy_Steganographic_Communication_Protocol}.
Assume that the conditions in
Section~\ref{Assumptions} hold, the complete canonical epoch
transcript is authenticated, and the receiver correctly maintains
its freshness state. For any PPT adversary $\mathcal A$, define
\begin{align}
\operatorname{Adv}^{\mathsf{MC\text{-}ATTACK}}_{\mathcal A}
(\lambda)
\coloneqq
\max\Big\{
&
\operatorname{Adv}^{\mathsf{MR}}_{\mathcal A}(\lambda),
\operatorname{Adv}^{\mathsf{MC\text{-}REPLAY}}_{\mathcal A}
(\lambda),
\nonumber\\[-0.2em]
&
\operatorname{Adv}^{\mathsf{MC\text{-}MitM}}_{\mathcal A}
(\lambda)
\Big\}.
\label{Theorem:3}
\end{align}
Then $\operatorname{Adv}^{\mathsf{MC\text{-}ATTACK}}_{\mathcal A}
(\lambda)$
\begin{align}
\leq
\max\Bigg\{
&
\frac{q_H(\lambda)}{2^\kappa}
+
\varepsilon_{\mathsf{param}}
+
\operatorname{negl}(\lambda),
\nonumber\\
&
\operatorname{Adv}^{\mathsf{EUF\text{-}CMA}}_{\mathsf{HMAC}}
(\mathcal B)
+
\frac{3Q(Q-1)}{2^{\nu+1}}
+
\delta_{\mathsf{state}}
+
\operatorname{negl}(\lambda),
\nonumber\\
&
\operatorname{Adv}^{\mathsf{EUF\text{-}CMA}}_{\mathsf{HMAC}}
(\mathcal B')
+
\frac{3Q(Q-1)}{2^{\nu+1}}
+ \delta_{\mathsf{state}}
+
\delta_{\mathsf{dec}}\\
&+
\operatorname{negl}(\lambda)
\Bigg\}.
\label{eq:combined-mc-security-bound}
\end{align}
Here, $\delta_{\mathsf{dec}}$ bounds accepted-transcript decoding or unmasking failure. Consequently, the protocol provides confidentiality and authenticated integrity when all terms on the right-hand side of \eqref{eq:combined-mc-security-bound} are negligible.
\end{theorem}

\begin{proof}[Proof sketch]
By Claims~\ref{Claim1} and~\ref{Claim2}, message recovery beyond
the prior-guessing baseline is bounded by key-recovery probability
and the masking-pseudorandomness gap. Claim~\ref{Claim3} shows
that an exact replay is rejected by the receiver's freshness
state, while a replay modified to appear fresh requires a MAC
forgery, accidental freshness reuse, or receiver-state failure.
By Claim~\ref{Claim4}, any accepted modification of
$\gamma_{1,e}$, $\gamma_{2,e}$, $s_e$, the epoch header, or a
freshness value either produces a new canonical transcript and
therefore an EUF-CMA forgery, or reuses a previously authenticated
transcript and reduces to the replay analysis. Taking the maximum
of the three task-specific advantages gives
\eqref{eq:combined-mc-security-bound}. Complete experiments,
reductions, and case analyses are provided in Supplementary
Appendix~F of the Supplementary Material.
\end{proof}

\section{Evaluation of Protocol Metrics: Methodology and Results}\label{sec:evaluation}

This section focuses on evaluating the hybrid system discussed in prior sections.

\subsection{Experimental Setup Specification}
\label{subsec:setup}

\paragraph{Computing environment}
All experiments were conducted on Ubuntu~24.04.1~LTS using an AMD Ryzen~5~7500F six-core processor, an NVIDIA GeForce RTX~3060 GPU, and 30\,GB of system memory. Experiment configurations, random seeds, realised payloads, valid cover--stego pair counts, and active embedding backends were recorded in the corresponding run logs.

\paragraph{Datasets and data partitioning}
Three publicly available narrative corpora obtained from Project
Gutenberg were used as benign textual sources:
\emph{Alice's Adventures in Wonderland},
\emph{Pride and Prejudice}, and
\emph{The Adventures of Sherlock Holmes}.\footnote{Project Gutenberg
eBook pages:
\emph{Alice's Adventures in Wonderland},
\url{https://www.gutenberg.org/ebooks/11};
\emph{Pride and Prejudice},
\url{https://www.gutenberg.org/ebooks/1342}; and
\emph{The Adventures of Sherlock Holmes},
\url{https://www.gutenberg.org/ebooks/1661}.}
For each corpus, three independent $70/30$ reference--test
partitions were generated using seeds
$\{42,43,44\}$. The reference partition was used exclusively for
keyed cover-stream selection, whereas the held-out partition was
reserved for naturalness evaluation.

Image experiments used BOSSBase~1.01, comprising $10{,}000$
grayscale images of size $512\times512$ pixels
\footnote{BOSSBase~1.01 contains 10{,}000 grayscale images in \texttt{.pgm} format (512\,$\times$\,512), as noted in the repository: \url{https://www.kaggle.com/datasets/zapak1010/bossbase-bows2}.}. The images were used at their native spatial
resolution to avoid introducing resampling artefacts. 

For each payload, cover and stego images were paired one-to-one.
The BOSSBase pairs were divided at the cover-identity level using a
fixed $70/15/15$ training/validation/test partition with seed
$20260605$. Thus, when all $10{,}000$ pairs were valid, the subsets
contained $7{,}000$, $1{,}500$, and $1{,}500$ pairs, respectively.
For the SRM+EC experiments, $9{,}996$ valid pairs produced
$6{,}997$ training, $1{,}499$ validation, and $1{,}500$ test pairs.
A cover and any stego-object derived from it were never assigned to
different subsets.

\paragraph{Cover-stream naturalness evaluation}
To evaluate the opacity of the cover parameters transmitted over
$C_1$ and $C_2$, we tested whether the keyed streams
$(\gamma_1,\gamma_2)$ remained statistically close to held-out
benign text. This experiment provides an empirical assessment of the
cover-parameter indistinguishability term
$\varepsilon_{\mathsf{cp\text{-}ind}}$ in the
$\textsf{CP-IND}$ game of Definition~\ref{def:cp-ind}.

For every corpus split, $V_{\mathsf{pri}}$ used HMAC-derived
deterministic indices to select two independent streams,
$\gamma_1$ and $\gamma_2$, each containing $1{,}000$ samples from
the reference partition. Each keyed stream was compared with a
$1{,}000$-sample subset drawn from the held-out partition of the
same corpus. A benign--benign reference was obtained by comparing an
independent benign reference subset with the corresponding held-out
subset.

Naturalness was evaluated using MAUVE, the divergence-frontier
measure introduced in \emph{MAUVE: Measuring the Gap Between Neural Text and Human Text using Divergence Frontiers}
\cite{pillutla2021mauve}; a fine-tuned RoBERTa-base binary
classifier \cite{liu2019roberta}; character- and word-level
$n$-gram Jensen--Shannon divergence for $n\in\{1,2,3\}$; and part-of-speech distribution differences measured using Cramér's~$V$. Together, these probes assess distributional, classifier-based, lexical, and syntactic separability. They provide empirical evidence concerning $\varepsilon_{\mathsf{cp\text{-}ind}}$ but do not constitute a
cryptographic upper bound on that term.

\paragraph{Embedding methods, baselines, and payloads}\label{Embedding methods}
The image-domain evaluation considered the proposed hybrid implementations and four established content-adaptive spatial steganographic methods: WOW, S-UNIWARD, HILL, and MiPOD \cite{Holub2012WOW,holub2014universal,fridrich2011steganalysis,setiadi2025maximizing}. All methods were evaluated at target payloads $p\in\{0.10,0.20,0.40\}\ \text{bpp}$. Each cover produced at most one stego-object for a given method and payload. Realised payload, extraction success, bit-error rate, and embedding runtime were logged for every attempted pair.

Message embedding was performed using syndrome-trellis codes (STCs) \footnote{pySTC by Daniel Lerch; available at \url{https://github.com/daniellerch/pySTC/blob/master/README.md}}. The implementation used the
\texttt{pySTC}/\texttt{pystcstego} Python interface,\footnote{Daniel
Lerch's pySTC implementation:
\url{https://github.com/daniellerch/pySTC}.}
which invokes an STC backend derived from the Binghamton
C++/MATLAB implementation.\footnote{Binghamton Digital Data
Embedding Laboratory STC implementation:
\url{https://dde.binghamton.edu/download/syndrome/}.}

The hybrid implementation applies a masking layer before
cover-adaptive image embedding, as specified in
Subsection~\ref{hybrid_model} and Algorithm~\ref{alg:HybridStegoModel}. The image embedding scheme is a QIM-style method with Fused distortion costs and Complex-Feature-aware refinement (\textbf{QIM-Fused-CF}), and follows Equations~\eqref{mask_a_hybrid}, \eqref{eq:qim-dither-mask}, and \eqref{eq:qim-dither-mask_a}.

The initial embedding cost is formed by log-domain fusion of
S-UNIWARD and MiPOD-inspired costs:
\[
    \rho_{i,\mathsf{base}}
    =
    \exp\!\left(
        \lambda\log\rho_{i,\mathsf{S}}
        +(1-\lambda)\log\rho_{i,\mathsf{M}}
    \right),
    \qquad \lambda=0.75.
\]
This base cost is refined using the complex-feature map
$\Psi(\mathbf{o})$, which assigns lower relative costs to textured, edge-rich, and high-residual regions. A conservative multiplicative factor $g(\psi_i(\mathbf{o}),\rho_{\mathsf{dit},i})$ then incorporates the cover-conditioned dither response. SLIC is used to form compact, locally coherent regions for the regional constraints \cite{achanta2012slic}, while local cost spreading follows the general motivation of HILL-style cost smoothing, which encourages spatially clustered low-cost embedding regions \cite{li2014new}. The resulting directional costs $\rho_{i,\mathsf{ours}}^{+}$ and $\rho_{i,\mathsf{ours}}^{-}$ are supplied to the STC encoder, which minimises the total additive distortion subject to the target payload. For a detailed description of the QIM construction, see Section~B-B of the Supplementary Material.

\paragraph{Steganalysis detectors and training protocol}
Detectability was evaluated using one classical feature-based
detector and two convolutional neural-network steganalyzers:
\textbf{SRM+EC}, \textbf{Ye-Net}, and
\textbf{Yedroudj-Net}. SRM+EC denotes the
$34{,}671$-dimensional \emph{Spatial Rich Model} followed by an
\emph{Ensemble Classifier}
\cite{Kodovsky2012Ensemble}. The SRM+EC experiments used $1{,}000$-dimensional random feature
subspaces. 

The CNN evaluation used Ye-Net
\cite{ye2017deep} and Yedroudj-Net
\cite{yedroudj2018yedroudj}. Both networks were trained for
80 epochs using stochastic gradient descent with momentum $0.9$, weight decay $10^{-4}$ + MultiStepLR, mini-batch size 12, and four data-loading workers. SGD is a standard choice for CNN-based steganalysis and is particularly suitable when training is intended to favour stable
generalisation rather than rapid adaptive optimisation
\cite{yedroudj2018yedroudj}.


\paragraph{Empirical linkability probes}\label{subsec:empirical_linkability}
The preceding evaluations assess the marginal appearance of the
cover streams $(\gamma_1,\gamma_2)$ and the stego-object $s$.
A multichannel warden may instead attempt to determine whether
objects observed over $C_1$, $C_2$, and $C_3$ originated from the
same protocol instance. We therefore evaluated the content-level
linkage represented by the $\textsf{CP-UNL}$ game defined in
Section~\ref{sec:security}.

For cover $i$, payload $p$, and generating cover-parameter pair $j$,
a linked tuple is $\mathcal{T}^{(1)}_{i,p,j}
    =
    \bigl(
        \gamma_1^{(j)},
        \gamma_2^{(j)},
        s_{i,p}^{(j)}
    \bigr).$

A shuffled tuple retains the same cover, payload, and stego-object
but replaces the candidate cover parameters with a different pair:
\[
    \mathcal{T}^{(0)}_{i,p,j,k}
    =
    \bigl(
        \gamma_1^{(k)},
        \gamma_2^{(k)},
        s_{i,p}^{(j)}
    \bigr),
    \qquad k\neq j.
\]
Accordingly, image identity, cover content, and payload cannot by
themselves determine the class label; successful classification
requires a reproducible statistical dependency between the candidate
textual parameters and the image generated using them.

The experiment used 12 independently generated
$(\gamma_1,\gamma_2)$ pairs and the same deterministic subset of
$1{,}000$ BOSSBase covers for every pair and payload. This produced $12\times1{,}000\times3=36{,}000$ stego-objects. One shuffled tuple was generated for each linked tuple,
yielding a balanced binary dataset containing $72{,}000$ observations.

Four adversarial views were evaluated: \emph{image-only}, \emph{text-only}, \emph{cross-only}, and \emph{text+image}. The image-only and text-only views serve as negative controls. The cross-only view contains explicit text--image interaction features, whereas the text+image view combines the image, parameter-stream, and cross-modal representations. Logistic regression
was used as a linear classification baseline
\cite{hastie2009elements}, while Extra Trees provided a nonlinear randomized ensemble baseline \cite{geurts2006extremely}.

Evaluation used five repeated cover-grouped $70/30$ train--test splits with seed $20260605$. All tuples derived from the same original cover were assigned to the same partition. In addition to binary linked-versus-shuffled classification, a closed-set ranking probe presented all 12 candidate cover-parameter pairs for each
held-out stego-object and measured Top-$1$, Top-$3$, Top-$5$, and mean reciprocal rank. These experiments provide an empirical proxy for the content-based component of $\varepsilon_{\mathsf{cp\text{-}unl}}$; they do not estimate scheduler- or metadata-based linkage.

\paragraph{Multichannel robustness and re-synchronisation}
Operational robustness was evaluated using a standalone Python
discrete-event simulator representing the three logical channels
$C_1$, $C_2$, and $C_3$. For epoch $e$, the simulator emits $H_e\Vert(\mathsf{nonce}_a\Vert\gamma_1) \quad\text{on }C_1,$    $H_e\Vert(\mathsf{nonce}_b\Vert\gamma_2)\quad\text{on }C_2,$ and
$ H_e\Vert(\mathsf{nonce}_c\Vert s_e\Vert\mathsf{MAC}_e) \quad\text{on }C_3.$
At the simulator level, the masked payload is abstracted as $b_e = m_e\oplus\gamma_{1,e} \oplus\gamma_{2,e} \oplus\rho_{\mathsf{dit},e}.$ The scheduler introduces controlled frame loss, burst loss, jitter, inter-channel skew, frame reordering, timeout events, heartbeat signalling, and optional erasure recovery. Each configuration was evaluated over $5{,}000$ epochs using ten independent seeds. The parameter sweep considered $r\in\{0,0.01,0.05\},\, J\in\{0,25,50,100\}\ \text{ms}, \Delta\in\{50,100,200\}\ \text{ms}, \, W\in\{3,5,10\},$
together with timeout factors in $\{2,3\}$ and heartbeat cadences
$h\in\{5,10,20\}$ epochs.

The receiver accepts an epoch only when its freshness, sliding-window admissibility, inter-channel alignment, and MAC verification checks all succeed. The experiment reports acceptance, recovery, re-synchronisation success, latency, throughput, integrity failures, and false accepts. It evaluates recovery and alignment under the admissible scheduler class and therefore informs the operational scheduler term $\varepsilon_{\mathsf{sched}}(\sigma^\star)$. Stego-object detectability and content-level linkability are evaluated
separately through the $\varepsilon_{\mathsf{emb}}$ and
$\varepsilon_{\mathsf{cp\text{-}unl}}$ probes described above.

\subsection{Cover-Stream Naturalness Evaluation}
\label{subsec:cover_stream_naturalness}

\begin{table*}[t]
\centering
\caption{Within-corpus cover-stream naturalness evaluation. Higher MAUVE indicates stronger distributional similarity to held-out benign text, while lower RoBERTa accuracy/AUC and lower POS Cramér's $V$ indicate weaker separability from benign text. Results are reported as mean~$\pm$~standard deviation across three random splits and two keyed cover streams.}
\label{tab:cover_stream_naturalness}
\small
\setlength{\tabcolsep}{3.5pt}
\renewcommand{\arraystretch}{1.15}
\begin{tabular}{lccccc}
\toprule
\textbf{Corpus} 
& \textbf{MAUVE $\uparrow$} 
& \textbf{Human MAUVE $\uparrow$} 
& \textbf{RoBERTa Acc. $\downarrow$} 
& \textbf{RoBERTa AUC $\downarrow$} 
& \textbf{POS Cramér's $V$ $\downarrow$} \\
\midrule
\emph{Alice in Wonderland} 
& $0.9585 \pm 0.0169$ 
& $0.9895 \pm 0.0060$ 
& $0.5770 \pm 0.0414$ 
& $0.6207 \pm 0.0156$ 
& $0.0263 \pm 0.0032$ \\

\emph{Pride and Prejudice} 
& $0.9676 \pm 0.0103$ 
& $0.9843 \pm 0.0039$ 
& $0.5238 \pm 0.0238$ 
& $0.5543 \pm 0.0306$ 
& $0.0349 \pm 0.0085$ \\

\emph{Sherlock Holmes} 
& $0.9729 \pm 0.0143$ 
& $0.9893 \pm 0.0004$ 
& $0.5379 \pm 0.0312$ 
& $0.5907 \pm 0.0141$ 
& $0.0321 \pm 0.0043$ \\
\bottomrule
\end{tabular}
\end{table*}

The results in table \ref{tab:cover_stream_naturalness} illustrates that the selected cover streams remains close to held-out benign text. The MAUVE scores are presented in the are consistently high and close to the human-vs-human baseline. For the RoBERTa scores, the separability remains near chance, while the POS Cramer's effect sizes are small. In terms of the theory discussed in Sec.~\ref{subsec:main-theorem}, these findings provide empirical support for a small cover-parameter indistinguishability term $\varepsilon_{\mathsf{cp\text{-}ind}}$, indicating that the keyed cover parameters $(\gamma_1,\gamma_2)$ are difficult to distinguish from benign within-corpus text under distributional, syntactic, and classifier-based probes.

\subsection{Stego-Object Embedding Fidelity and Payload Realisation}
\label{subsec:embedding_fidelity}
Table~\ref{tab:embedding_fidelity} summarise embedding fidelity under matched payloads outlined in \ref{subsec:setup}. All methods realise the intended payload rates, and every valid \texttt{pySTC} embedding is extracted with zero BER. The proportion of valid cover/stego pairs remains high across the experiments, exceeding (97$\%$) in all cases and reaching (100$\%$) for MiPOD at (0.1) bpp.

However, as the payload increases, fidelity decreases as expected. For the hybrid Ours-QIM, it decreases from (64.21) dB to (57.68) dB. This trend is accompanied by increasing change rates, confirming the expected payload distortion trade-off. Among the baselines, MiPOD achieves the highest PSNR at (0.1) bpp, whereas S-UNIWARD achieves the highest PSNR at (0.2) and (0.4) bpp. Nevertheless, both proposed variants remain close to these strongest baselines and preserve very high structural similarity, with SSIM remaining above (0.9995) even at (0.4) bpp.

Overall, the results show that the hybrid scheme proposed in \ref{Hybrid_stego_model} and \ref{Hybrid_Entropy_Steganographic_Communication_Protocol} support reliable payload recovery and high-fidelity stego-object generation on ($C_3$) under the same payload conditions as the baselines. These findings establish implementation correctness and image utility. 


\begin{table*}[!t]
\centering
\caption[Embedding fidelity and payload realisation.]
{Embedding fidelity and payload realisation for the hybrid variants and spatial-domain baselines.}
\label{tab:embedding_fidelity}
\small
\setlength{\tabcolsep}{4.2pt}
\renewcommand{\arraystretch}{1.15}
\begin{threeparttable}
\begin{tabular}{
@{}
ll
c
c
c
c
c
@{}}
\toprule
\textbf{Method} &
\textbf{Payload (bpp)} &
\textbf{Valid / Total} &
\textbf{Actual Payload (bpp)} &
\textbf{Change Rate (\%)} &
\textbf{PSNR (dB)} &
\textbf{SSIM} \\
\midrule

\multirow{3}{*}{Ours-QIM}
& 0.10 &  9904 / 10000 & $0.1000 \pm 0.0000$ & $ 2.477 \pm 0.249$ & $64.21 \pm 0.43$ & $0.999947 \pm 0.000050$ \\
& 0.20 &  9896 / 10000 & $0.2000 \pm 0.0000$ & $ 5.239 \pm 0.482$ & $60.96 \pm 0.40$ & $0.999849 \pm 0.000125$ \\
& 0.40 &  9846 / 10000 & $0.4000 \pm 0.0000$ & $11.140 \pm 0.915$ & $57.68 \pm 0.36$ & $0.999558 \pm 0.000313$ \\
\addlinespace[2pt]

\multirow{3}{*}{WOW}
& 0.10 &  9868 / 10000 & $0.1000 \pm 0.0000$ & $ 2.111 \pm 0.340$ & $64.93 \pm 0.62$ & $0.999704 \pm 0.000087$ \\
& 0.20 &  9896 / 10000 & $0.2000 \pm 0.0000$ & $ 4.545 \pm 0.612$ & $61.59 \pm 0.55$ & $0.999381 \pm 0.000168$ \\
& 0.40 &  9871 / 10000 & $0.4000 \pm 0.0000$ & $10.136 \pm 1.227$ & $58.10 \pm 0.51$ & $0.998686 \pm 0.000373$ \\
\addlinespace[2pt]

\multirow{3}{*}{S-UNIWARD}
& 0.10 &  9999 / 10000 & $0.1000 \pm 0.0000$ & $ 1.685 \pm 0.185$ & $65.89 \pm 0.47$ & $0.999939 \pm 0.000049$ \\
& 0.20 &  9996 / 10000 & $0.2000 \pm 0.0000$ & $ 3.671 \pm 0.367$ & $62.50 \pm 0.43$ & $0.999840 \pm 0.000114$ \\
& 0.40 &  9996 / 10000 & $0.4000 \pm 0.0000$ & $ 8.227 \pm 0.711$ & $58.99 \pm 0.37$ & $0.999566 \pm 0.000273$ \\
\addlinespace[2pt]

\multirow{3}{*}{HILL}
& 0.10 &  9829 / 10000 & $0.1000 \pm 0.0000$ & $ 1.922 \pm 0.305$ & $65.35 \pm 0.68$ & $0.999932 \pm 0.000055$ \\
& 0.20 &  9825 / 10000 & $0.2000 \pm 0.0000$ & $ 4.197 \pm 0.607$ & $61.95 \pm 0.62$ & $0.999826 \pm 0.000126$ \\
& 0.40 &  9740 / 10000 & $0.4000 \pm 0.0000$ & $ 9.380 \pm 1.226$ & $58.45 \pm 0.57$ & $0.999540 \pm 0.000299$ \\
\addlinespace[2pt]

\multirow{3}{*}{MiPOD}
& 0.10 & 10000 / 10000 & $0.1000 \pm 0.0000$ & $ 1.604 \pm 0.173$ & $66.10 \pm 0.44$ & $0.999929 \pm 0.000051$ \\
& 0.20 &  9997 / 10000 & $0.2000 \pm 0.0000$ & $ 3.764 \pm 0.460$ & $62.40 \pm 0.50$ & $0.999828 \pm 0.000120$ \\
& 0.40 &  9994 / 10000 & $0.4000 \pm 0.0000$ & $ 9.234 \pm 1.216$ & $58.51 \pm 0.55$ & $0.999556 \pm 0.000298$ \\
\bottomrule
\end{tabular}
\begin{tablenotes}
\footnotesize
\item \textit{Notes:} Results are computed over valid cover--stego pairs using true \texttt{pySTC} embedding and successful extraction. Values are reported as mean~$\pm$~standard deviation. All valid pairs achieved zero bit-error rate (BER). ``Valid / Total'' reports the number of successfully embedded and extracted image pairs out of 10,000 attempted covers.
\end{tablenotes}
\end{threeparttable}
\end{table*}

\subsection{Stego-Object Detectability Results}
\label{subsec:detectability_results}

To empirically assess the embedding term $\epsilon_{\mathrm{emb}}$ in
Section~\ref{sec:security}, we evaluate the image stego-object $s$
using SRM+EC, Ye-Net, and Yedroudj-Net. The adversary in this
experiment observes only the stego-object on the image channel; it
does not obtain the stego-key, the synthesised  cover messages, or the
complete multichannel transcript required for message recovery.

Table~\ref{tab:detectors_payloads_test} and
Fig.~\ref{Figure_4} show that detectability is strongly
payload-dependent. At 0.10 bpp, all detectors operate close to random
classification. SRM+EC achieves 51.37\% accuracy and an AUC of
0.5444, while Ye-Net and Yedroudj-Net obtain AUC values of 0.5017
and 0.5042, respectively. This indicates that, at low payload, the
proposed embedding introduces only a weak statistical footprint in the
observable stego-object. A similar pattern is observed at 0.20 bpp:
SRM+EC remains near chance with 50.53\% accuracy and AUC 0.5487,
Ye-Net obtains AUC 0.5424, and Yedroudj-Net rises only moderately
to AUC 0.5678. These results support the security intuition that the
cover-conditioned QIM dither, fused distortion design, and region-aware
embedding help keep $\mathrm{Law}(S_i)$ close to the corresponding
cover distribution $D_i$ in the low-to-moderate payload regime.

The detection error probability $P_e$ also provides a direct link
between the empirical detector and the embedding-distance term. For an
optimal equal-prior cover--stego test, $P_e^\star=\frac{1}{2}(1-\mathrm{TV}(D_i,\mathrm{Law}(S_i)))$, so $\mathrm{TV}(D_i,\mathrm{Law}(S_i))=1-2P_e^\star$.
Thus, larger $P_e$ implies smaller cover--stego separation and,
therefore, a smaller empirical contribution to
$\epsilon_{\mathsf{emb}}$. In the SRM+EC results,
$1-2P_e$ is approximately $0.1340$, $0.1566$, and $0.5906$ at
0.10, 0.20, and 0.40 bpp, respectively, confirming that the
embedding footprint is weak at low payloads but becomes substantially
more visible at 0.40 bpp.

The behaviour changes at 0.40 bpp. SRM+EC reaches 79.27\%
accuracy and AUC 0.8860, while Yedroudj-Net achieves the strongest
detection performance with 82.69\% accuracy and AUC 0.9302. This
shows that, once the payload becomes large, the accumulated embedding
changes become visible to both classical residual co-occurrence features
and learned CNN features. In contrast, Ye-Net remains comparatively
weak at this payload, with 55.70\% accuracy and AUC 0.5811, suggesting
that this detector is less sensitive to the particular artefacts produced by
QIM-Fused-CF under the present training configuration.
Therefore, the strongest high-payload evidence comes from SRM+EC and
Yedroudj-Net rather than from Ye-Net.

Overall, the results refine rather than contradict the theoretical security
discussion in Section~\ref{sec:security}. The proposed construction does
not imply payload-independent undetectability; instead, the composition
bound predicts that the warden's advantage depends partly on the
embedding footprint $\epsilon_{\mathrm{emb}}$. The near-random AUC values
at 0.10 and 0.20 bpp indicate that this footprint remains small for lower
payloads, whereas the sharp increase at 0.40 bpp shows that
$\epsilon_{\mathrm{emb}}$ grows as more information is embedded into the
cover. Thus, QIM-Fused-CF provides its strongest
security--capacity trade-off at 0.10--0.20 bpp, while 0.40 bpp should be
interpreted as a high-capacity operating point with substantially greater
detectability. Detailed detector-wise comparisons and complete 80-epoch learning trajectories are provided in Supplementary Figs.~5--7 in Appendix~H; they confirm persistent near-chance discrimination at 0.10~bpp, weak emerging evidence at 0.20~bpp, and a pronounced increase in $\epsilon_{\mathrm{emb}}$ at 0.40~bpp, while also revealing the configuration-specific collapse of Ye-Net after its early optimum.

\begin{table*}[!t]
\centering
\caption[Stego-object detectability across payloads.]
{Stego-object detectability across payloads using SRM+EC and CNN-based steganalyzers.}
\label{tab:detectors_payloads_test}
\small
\setlength{\tabcolsep}{4.3pt}
\begin{threeparttable}
\begin{tabular}{
@{}
l
S[table-format=1.2]
S[table-format=4.0]
S[table-format=2.2]
S[table-format=2.2]
S[table-format=1.4]
S[table-format=2.2]
S[table-format=2.2]
S[table-format=1.4]
S[table-format=1.4]
c
@{}}
\toprule
\textbf{Detector} &
{\textbf{Payload (bpp)}} &
{\textbf{Pairs}} &
{\textbf{ACC} (\%)} &
{\textbf{BACC} (\%)} &
{\textbf{AUC}} &
{\textbf{EER} (\%)} &
{\boldmath$P_e$ (\%)} &
{\textbf{MCC}} &
{\textbf{Best epoch}} &
\textbf{Subspace}\\
\midrule
\multicolumn{10}{@{}l}{\textit{Classical steganalysis}} \\
SRM+EC       & 0.10 & 1500 & 51.37 & 51.37 & 0.5444 & 46.93 & 43.30 & 0.0275 & {--} & 1000 \\
SRM+EC       & 0.20 & 1500 & 50.53 & 50.53 & 0.5487 & 47.50 & 42.17 & 0.0108 & {--} & 1000 \\
SRM+EC       & 0.40 & 1500 & 79.27 & 79.27 & 0.8860 & 20.93 & 20.47 & 0.5854 & {--} & 1000 \\
\addlinespace[2pt]
\multicolumn{10}{@{}l}{\textit{CNN-based steganalysis}} \\
Ye-Net       & 0.10 & 2000 & 50.05 & {--} & 0.5017 & 49.94 &  {--} & {--} & 80 & {--}\\
Ye-Net       & 0.20 & 2000 & 50.30 & {--} & 0.5424 & 49.69 &  {--} & {--} & 80 & {--} \\
Ye-Net       & 0.40 & 2000 & 55.70 & {--} & 0.5811 & 44.29 &  {--} & {--} & 12 & {--} \\
Yedroudj-Net & 0.10 & 2000 & 50.30 & {--} & 0.5042 & 49.69 & {--}  & {--} & 59 & {--} \\
Yedroudj-Net & 0.20 & 2000 & 54.30 & {--} & 0.5678 & 45.69 & {--}  & {--} & 79 & {--} \\
Yedroudj-Net & 0.40 & 2000 & 82.69 & {--} & 0.9302 & 17.30 & {--}  & {--} & 71 & {--} \\
\bottomrule
\end{tabular}
\begin{tablenotes}
\footnotesize
\item \textit{Notes:} ACC = accuracy; BACC = balanced accuracy; AUC = area under the ROC curve; EER = equal error rate; $P_e$ = detection error probability; MCC = Matthews correlation coefficient. Missing values indicate metrics not reported by the corresponding detector run. For SRM+EC, ``Best epoch'' is not applicable.
\end{tablenotes}
\end{threeparttable}
\end{table*}

\begin{figure}[!t]
    \centering
    \includegraphics[width=1.0\linewidth]{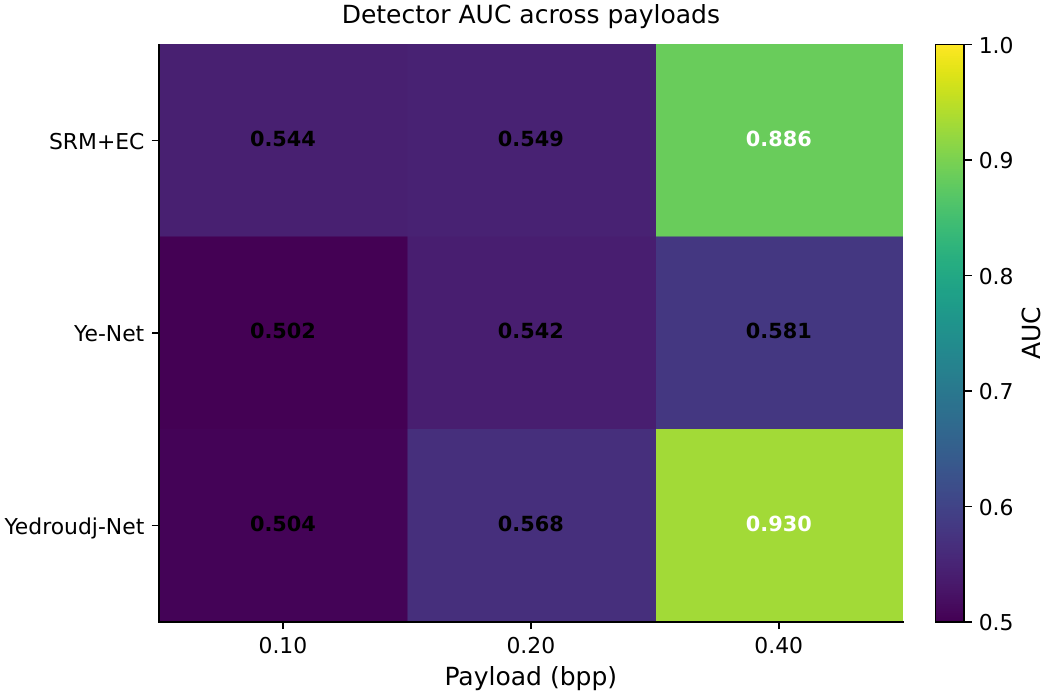}
    \caption{Detector AUC across payloads for SRM+EC, Ye-Net, and Yedroudj-Net. Near-0.5 AUC values indicate weak cover--stego discrimination, whereas higher values indicate stronger detectability; the strongest detection occurs at 0.40 bpp.}
    \label{Figure_4}
\end{figure}

\subsection{Empirical Linkability Probes}
\label{sec:empirical_linkability}


\begin{table*}[!t]
\centering
\caption[Linked-versus-shuffled linkability classification.]
{Linked-versus-shuffled linkability classification under grouped cover-level evaluation.}
\label{tab:linkability_binary_results}
\scriptsize
\setlength{\tabcolsep}{2.8pt}
\renewcommand{\arraystretch}{1.12}
\begin{threeparttable}
\begin{tabular}{
@{}
l
S[table-format=1.2]
c
c
c
c
c
c
c
@{}}
\toprule
\textbf{View} &
{\textbf{Payload}} &
\textbf{ACC} &
\textbf{BACC} &
\textbf{AUC} &
\textbf{EER} &
{\boldmath$P_e$} &
{\boldmath$1-2P_e$} &
\textbf{MCC} \\
&
{\textbf{(bpp)}} &
\textbf{(\%)} &
\textbf{(\%)} &
&
\textbf{(\%)} &
\textbf{(\%)} &
&
\\
\midrule
\multirow{3}{*}{Cross only}
& 0.10 & $49.92 \pm 0.14$ & $49.92 \pm 0.14$ & $0.5002 \pm 0.0018$ & $50.06 \pm 0.12$ & $50.08 \pm 0.14$ & $-0.0016 \pm 0.0029$ & $-0.0016 \pm 0.0029$ \\
& 0.20 & $49.77 \pm 0.22$ & $49.77 \pm 0.22$ & $0.4979 \pm 0.0019$ & $50.28 \pm 0.25$ & $50.23 \pm 0.22$ & $-0.0046 \pm 0.0044$ & $-0.0046 \pm 0.0044$ \\
& 0.40 & $49.87 \pm 0.29$ & $49.87 \pm 0.29$ & $0.5000 \pm 0.0011$ & $50.18 \pm 0.26$ & $50.12 \pm 0.29$ & $-0.0025 \pm 0.0057$ & $-0.0025 \pm 0.0057$ \\
\addlinespace[2pt]

\multirow{3}{*}{Image only}
& 0.10 & $50.00 \pm 0.00$ & $50.00 \pm 0.00$ & $0.5000 \pm 0.0000$ & $50.00 \pm 0.00$ & $50.00 \pm 0.00$ & $0.0000 \pm 0.0000$ & $0.0000 \pm 0.0000$ \\
& 0.20 & $50.00 \pm 0.00$ & $50.00 \pm 0.00$ & $0.5000 \pm 0.0000$ & $50.00 \pm 0.00$ & $50.00 \pm 0.00$ & $0.0000 \pm 0.0000$ & $0.0000 \pm 0.0000$ \\
& 0.40 & $50.00 \pm 0.00$ & $50.00 \pm 0.00$ & $0.5000 \pm 0.0000$ & $50.00 \pm 0.00$ & $50.00 \pm 0.00$ & $0.0000 \pm 0.0000$ & $0.0000 \pm 0.0000$ \\
\addlinespace[2pt]

\multirow{3}{*}{Text--image}
& 0.10 & $49.93 \pm 0.11$ & $49.93 \pm 0.11$ & $0.5006 \pm 0.0025$ & $50.11 \pm 0.17$ & $50.07 \pm 0.11$ & $-0.0014 \pm 0.0022$ & $-0.0014 \pm 0.0022$ \\
& 0.20 & $49.81 \pm 0.17$ & $49.81 \pm 0.17$ & $0.4966 \pm 0.0025$ & $50.17 \pm 0.20$ & $50.19 \pm 0.17$ & $-0.0038 \pm 0.0034$ & $-0.0038 \pm 0.0034$ \\
& 0.40 & $49.99 \pm 0.16$ & $49.99 \pm 0.16$ & $0.5013 \pm 0.0024$ & $49.95 \pm 0.18$ & $50.01 \pm 0.16$ & $-0.0002 \pm 0.0032$ & $-0.0002 \pm 0.0032$ \\
\addlinespace[2pt]

\multirow{3}{*}{Text only}
& 0.10 & $50.00 \pm 0.00$ & $50.00 \pm 0.00$ & $0.5000 \pm 0.0000$ & $50.00 \pm 0.00$ & $50.00 \pm 0.00$ & $0.0000 \pm 0.0000$ & $0.0000 \pm 0.0000$ \\
& 0.20 & $50.00 \pm 0.00$ & $50.00 \pm 0.00$ & $0.5000 \pm 0.0000$ & $50.00 \pm 0.00$ & $50.00 \pm 0.00$ & $0.0000 \pm 0.0000$ & $0.0000 \pm 0.0000$ \\
& 0.40 & $50.00 \pm 0.00$ & $50.00 \pm 0.00$ & $0.5000 \pm 0.0000$ & $50.00 \pm 0.00$ & $50.00 \pm 0.00$ & $0.0000 \pm 0.0000$ & $0.0000 \pm 0.0000$ \\
\bottomrule
\end{tabular}
\begin{tablenotes}
\footnotesize
\item \textit{Notes:} Values are reported as mean~$\pm$~standard deviation across repeated grouped cover-level evaluations. ACC = accuracy; BACC = balanced accuracy; AUC = area under the ROC curve; EER = equal error rate; $P_e$ = detection error probability; MCC = Matthews correlation coefficient. The term $1-2P_e$ is the total-variation-style distinguishability proxy induced by the binary error probability; values close to zero indicate no measurable linkability beyond chance.
\end{tablenotes}
\end{threeparttable}
\end{table*}

\begin{figure*}[t]
    \centering
    \renewcommand{\thesubfigure}{\roman{subfigure}}

    \subfloat[Grouped-split stability of the Extra-Trees warden.
    \label{fig:linkability_split_stability}]{
        \includegraphics[
            width=0.455\textwidth
        ]{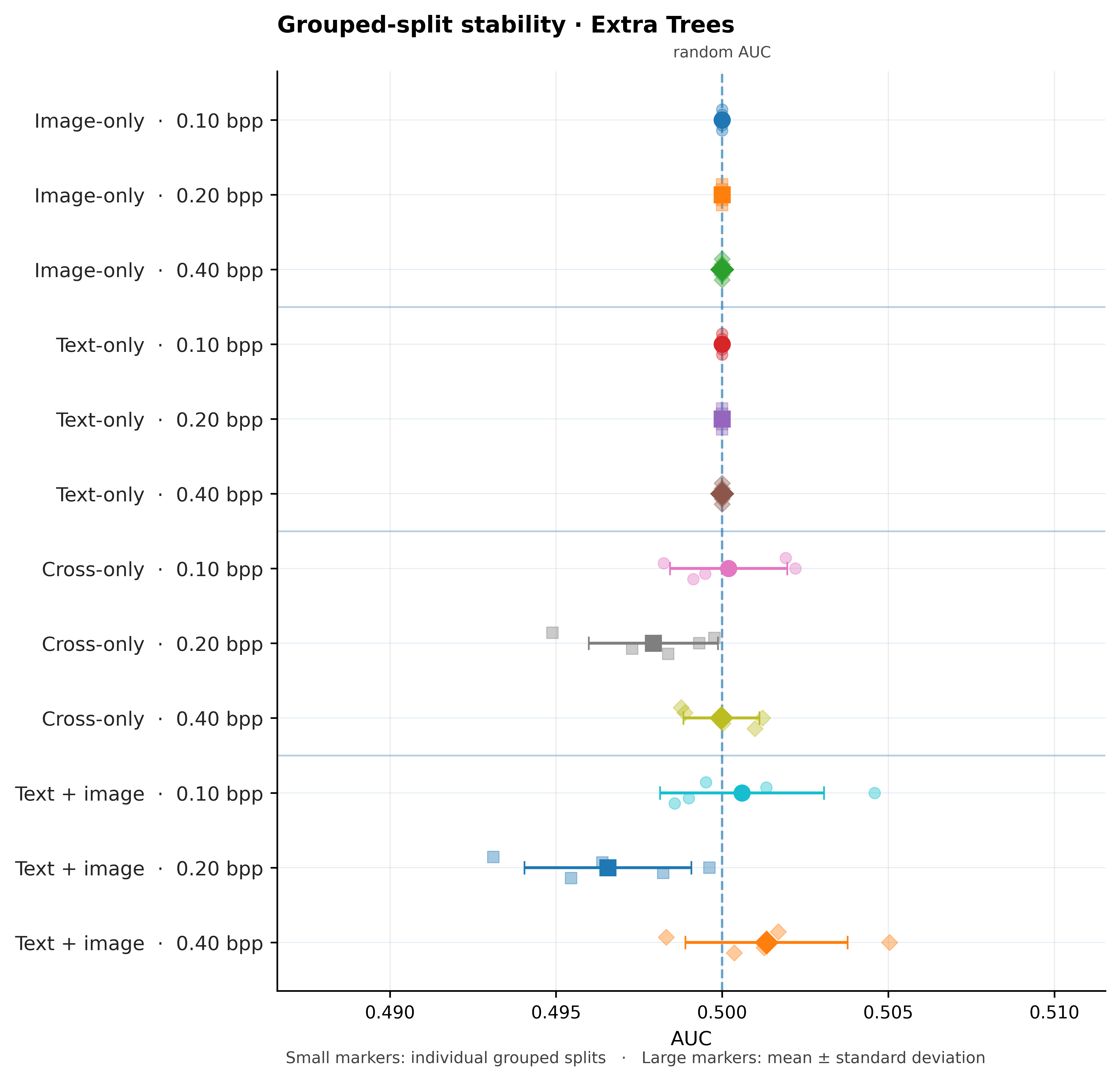}
    }%
    \hfill
    \subfloat[Empirical linkability advantage across payloads and
    adversary views.
    \label{fig:linkability_advantage}]{
        \includegraphics[
            width=0.50\textwidth
        ]{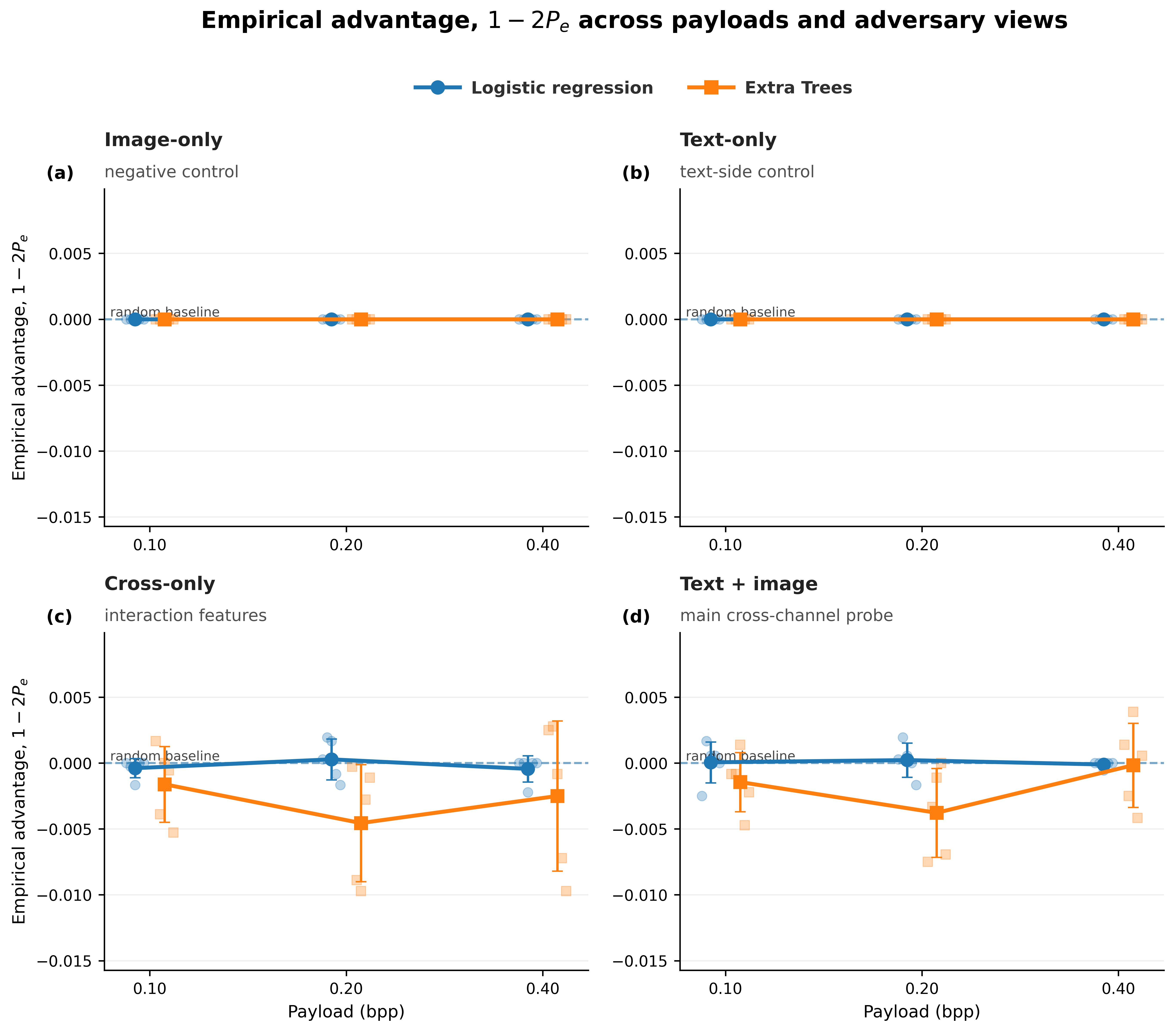}
    }

    \caption{Empirical linkability results under repeated grouped
    cover-level evaluation. Subfigure~(i) shows the split-level
    stability of the Extra-Trees warden, where the small markers
    represent individual grouped splits and the large markers represent
    the mean with standard deviation. Subfigure~(ii) compares the signed
    empirical quantity $1-2P_e$ obtained by Logistic Regression and
    Extra Trees under the image-only, text-only, cross-only, and
    text--image adversary views. Results concentrated around the random
    baselines indicate no consistent content-level cross-channel
    linkage.}
    \label{fig:linkability_combined_results}
\end{figure*}

Table~\ref{tab:linkability_binary_results} and
Fig.~\ref{fig:linkability_combined_results} show that all four
views remain at, or extremely close to, chance performance. As expected
from the construction of the matched tuples, the image-only view gives
exactly $50.00\%$ accuracy, $\mathrm{AUC}=0.5000$,
$\mathrm{EER}=50.00\%$, and $\mathrm{MCC}=0$ at every payload. This
result is an important implementation check: the linked and shuffled
examples contain the same stego-object, and therefore no image-only
classifier should be able to distinguish their labels. The text-only
view produces the same chance-level results because each candidate
cover-parameter pair occurs symmetrically in the linked and shuffled
classes. Hence, the classifier cannot exploit the identity, corpus, or
marginal statistics of an individual $(\gamma_1,\gamma_2)$ pair as a
shortcut for the linkage label.

More importantly, the cross-only and text+image views also fail to
produce meaningful separability. For the cross-only warden, the AUC is
$0.5002\pm0.0018$, $0.4979\pm0.0019$, and
$0.5000\pm0.0011$ at $0.10$, $0.20$, and $0.40$~bpp,
respectively. The corresponding balanced accuracies range only from
$49.77\%$ to $49.92\%$, while the MCC remains between
$-0.0046$ and $-0.0016$. Thus, the explicit interaction features do not
expose a repeatable compatibility signature between the cover
parameters and their associated stego-object.

The strongest adversarial view, which combines the text, image, and
cross-modal features, reaches AUC values of
$0.5006\pm0.0025$, $0.4966\pm0.0025$, and
$0.5013\pm0.0024$ across the three payloads. Its accuracy and balanced
accuracy remain within $0.19$ percentage points of $50\%$, its EER
remains close to $50\%$, and its MCC remains close to zero. The small
positive AUC at $0.40$~bpp does not constitute a payload-dependent
linkage trend: it is only $0.0013$ above chance and is comparable to the
observed split-to-split standard deviation. Conversely, the
below-$0.5$ AUC at $0.20$~bpp indicates a small reversal in score
ordering rather than useful adversarial information. A warden could
invert such a score, but the resulting absolute displacement from
chance would still be negligible in practical terms.

The grouped-split results in
Fig.~\ref{fig:linkability_split_stability} further show that these
findings are stable across cover-level partitions. The image-only and
text-only observations collapse exactly onto the random baseline. The
cross-only and text+image results display slightly greater variability,
particularly at $0.20$~bpp, but their means and error intervals remain
centred near $\mathrm{AUC}=0.5$. There is no consistent direction of
departure across payloads: some splits fall marginally below chance and
others marginally above it. This absence of a reproducible direction is
important because a genuine pair-specific footprint would be expected
to produce a consistently oriented score across independently held-out
cover groups.

The corresponding empirical \textsf{CP-UNL} advantage is
\begin{equation}
\widehat{\operatorname{Adv}}_{\mathsf{CP\text{-}UNL}}
(\mathcal A)
=
\left|\frac{1}{2}-\widehat P_e\right|
=
\frac{1}{2}\left|1-2\widehat P_e\right|.
\label{eq:empirical_cp_unl_adv}
\end{equation}
The signed quantity $1-2\widehat P_e$ reported in
Table~\ref{tab:linkability_binary_results} remains between
$-0.0046$ and $0$ for the selected Extra-Trees warden. Negative values
mean only that the fixed decision rule performs slightly below random
guessing; they do not represent negative leakage. After taking the
absolute value required by the CP-UNL definition, the largest observed
empirical advantage is approximately $0.0023$. This should be
interpreted as the advantage of the particular tested feature-based
warden, rather than as a universal cryptographic upper bound over all
PPT adversaries.

These findings complement, but should not be conflated with, the preceding \textsf{CP-IND} and stego-detectability experiments. The cover-stream naturalness evaluation assesses whether $(\gamma_1,\gamma_2)$ resemble benign text and therefore provides evidence concerning $\epsilon_{\mathsf{cp\text{-}ind}}$. The stego-detectability evaluation assesses the marginal embedding footprint on $C_3$ and therefore provides evidence concerning $\epsilon_{\mathsf{emb}}$. By contrast, the present experiment asks whether the joint observation $(\gamma_1,\gamma_2,s)$ exposes a same-session dependency and therefore provides an empirical proxy for the content-based component of $\epsilon_{\mathsf{cp\text{-}unl}}$. Accordingly, the results are consistent with the refined composition
bound in Section~\ref{sec:security}. For the candidate-ranking experiment, see Appendix~G of the Supplementary Material.

Overall, the agreement between the binary classification and
twelve-candidate ranking results provides consistent evidence that the
selected cover-parameter pair does not leave a recoverable,
pair-specific signature in the resulting stego-object. Increasing the
payload from $0.10$ to $0.40$~bpp does not materially increase this
linkability. Within the tested adversarial feature spaces and
cover-grouped evaluation, the observations therefore support the
content-level unlinkability objective of CP-UNL and strengthen the
empirical evidence underlying the multichannel composition argument.

\subsection{Multichannel Robustness and Re-synchronisation}
\label{subsec:multichannel_robustness}

Tables~\ref{tab:qim_resync_robustness}--\ref{tab:qim_resync_latency} and Figure~\ref{fig:qim_resync_combined} evaluate the practical robustness of the multichannel synchronisation mechanism specified in Algorithm~\ref{alg:mc-sync}. While Section~\ref{proto_Des} defines the three transmitted objects, defined in Equation~\eqref{eq:trans}, this experiment tests how reliably these objects remain aligned under the admissible scheduler class $\Sigma$ when loss, jitter, skew, timeout and reordering are introduced.

The results support the robustness behaviour anticipated by the windowed re-synchronisation analysis in Section~\ref{subsec:complexity-robustness} and Lemma~\ref{lem:resync}. In particular, the optional erasure-recovery mechanism substantially improves epoch acceptance. Under no injected loss, acceptance increases from $81.87\%$ without ECC to $97.13\%$ with ECC; under $5\%$ loss, it increases from $53.41\%$ to $78.31\%$. This improvement is explained by the recovery rate: without ECC, incomplete epochs cannot be reconstructed, whereas with ECC the receiver recovers between $51.96\%$ and $61.29\%$ of partially observed epochs across the tested loss settings.

The re-synchronisation trend is also consistent with Algorithm~\ref{alg:mc-sync}. As loss increases, more epochs require receiver-side correction through the sliding window $W$, timeout $\tau$, skew tolerance $\Delta$, and heartbeat alignment. At the same time, the integrity checks remain effective: the integrity-failure rate is $0\%$ and no false accepts are observed in any tested condition. This empirically supports the protocol’s nonce freshness, epoch-header binding, and MAC-verification arguments, since stale, incomplete, or unauthenticated channel combinations are rejected rather than accepted as valid epochs.

The latency results \ref{tab:qim_resync_latency} show that this robustness is achieved with bounded operational cost. As jitter increases from $0$ to $100$ ms, tail latency increases as expected, but ECC consistently reduces accepted-epoch delay and improves throughput. At $100$ ms jitter, P95 latency decreases from $184.94$ ms without ECC to $159.27$ ms with ECC, while mean latency decreases from $95.12$ ms to $56.76$ ms. Throughput also remains higher with ECC, staying around $8.7$--$8.9$ accepted epochs/s compared with approximately $6.7$--$7.1$ epochs/s without ECC.

\begin{table}[t]
\centering
\caption{Active replay/tamper robustness of the QIM-style multichannel simulator. Results are reported over simulator runs.}
\label{tab:qim_resync_active_interference}
\setlength{\tabcolsep}{5pt}
\renewcommand{\arraystretch}{1.15}
\resizebox{\linewidth}{!}{%
\begin{tabular}{@{}llrrrr@{}}
\toprule
Loss $r$ & ECC & MAC Failures & Replay Reject. & Stale Reject. & False Accepts \\
\midrule
$1\%$ & Off & 5 & 123 & 12 & 0 \\
$1\%$ & On & 0 & 96 & 14 & 0 \\
$5\%$ & Off & 5 & 120 & 12 & 0 \\
$5\%$ & On & 1 & 103 & 14 & 0 \\
\bottomrule
\end{tabular}%
}
\end{table}
Table~\ref{tab:qim_resync_active_interference} reports the behaviour of Algorithm~\ref{alg:mc-sync} under injected replay and tamper attempts. Across both loss settings, replayed frames are rejected by the epoch-window and nonce-freshness checks, while modified $C_3$ frames trigger MAC failures. Most importantly, no false accepts are observed in any condition. This supports the integrity and replay-resilience arguments in the protocol: the header binding $H_e=(e,N_e)$ aligns the three channel observations to a unique epoch, while Equations ~\eqref{MAC} and \eqref{MAC Verifi} prevents altered stego objects from being accepted. Thus, even when the scheduler introduces loss and the adversary injects stale or tampered frames, the receiver rejects invalid epochs rather than outputting an incorrect message fragment.

Overall, these results provide empirical support for the correctness and re-synchronisation claims of Algorithm~\ref{alg:mc-sync}. In the notation of Theorem~\ref{thm:composition}, the experiment concerns the scheduler component $\varepsilon_{\mathsf{sched}}(\sigma^\star)$: it shows that the transport discipline remains stable under the tested admissible perturbations while preserving epoch integrity. It does not by itself establish stego-object indistinguishability; the empirical assessment of $\varepsilon_{\mathsf{emb}}$ remains part of the separate stego-object detectability evaluation. For masking ablations, scheduling footprint, runtime efficiency, and communication overhead, see Sections~G-B--G-D of the Supplementary Material.

\begin{figure*}[!ht]
    \centering
    \includegraphics[width=\textwidth]{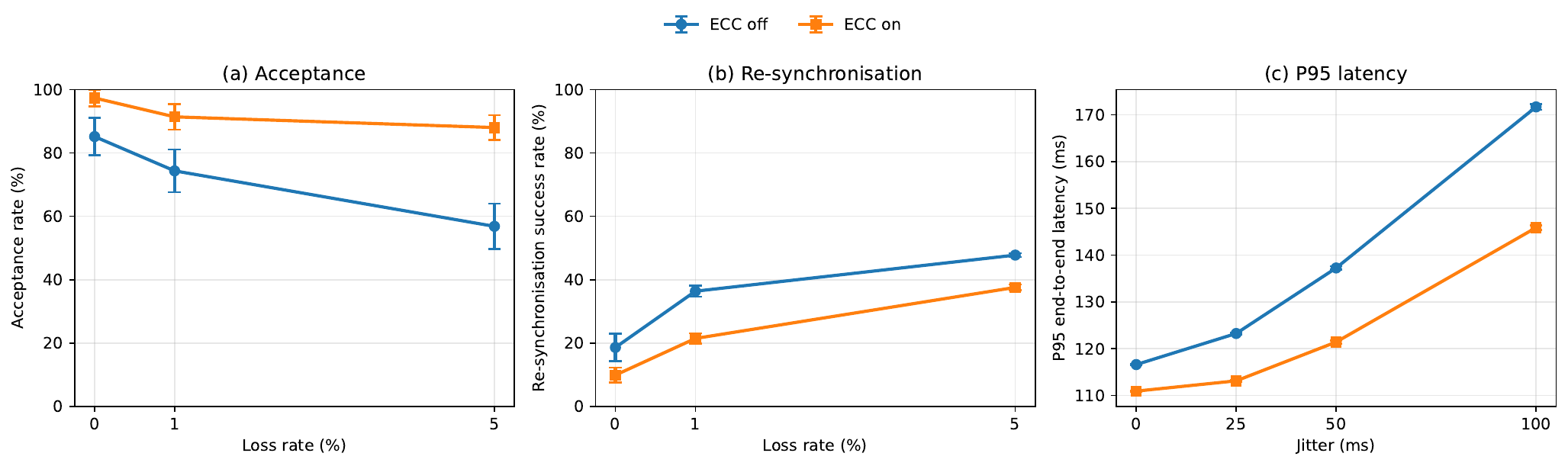}
    \caption{Multichannel robustness of the QIM-style hybrid protocol under controlled channel impairments. Panels show: (a) epoch acceptance rate under increasing channel loss, (b) re-synchronisation success rate under increasing channel loss, and (c) 95th-percentile end-to-end latency under increasing jitter. Results are shown for the configuration \(W=5\), \(h=10\), \(\Delta=100\) ms, and timeout \(2\Delta\); error bars denote 95\% confidence intervals across simulator runs.}
    \label{fig:qim_resync_combined}
\end{figure*}

\begin{table*}[!t]
\centering
\caption[QIM-style multichannel robustness under channel loss. Results are reported as mean~$\pm$~standard deviation over simulator runs, except false accepts, which are reported as total counts.]
{QIM-style multichannel robustness under channel loss. Results are reported as mean~$\pm$~standard deviation over simulator runs, except false accepts, which are reported as total counts.}
\label{tab:qim_resync_robustness}
\small
\setlength{\tabcolsep}{4.5pt}
\renewcommand{\arraystretch}{1.15}
\begin{threeparttable}
\begin{tabular}{
@{}
l
c
r
c
c
c
c
r
@{}}
\toprule
\textbf{Loss \boldmath$r$} &
\textbf{ECC} &
\textbf{Runs} &
\textbf{Acceptance (\%)} &
\textbf{Recovery (\%)} &
\textbf{Re-sync success (\%)} &
\textbf{Integrity fail (\%)} &
\textbf{False accepts} \\
\midrule
\multirow{2}{*}{$0\%$}
& Off & 4,460 & $81.87 \pm 31.00$ & $ 0.00 \pm 0.00$  & $20.93 \pm 22.81$ & $0.0000 \pm 0.0000$ & 0 \\
\multirow{2}{*}{$0\%$}& On  & 4,451 & $97.13 \pm 11.74$ & $61.29 \pm 7.57$ & $15.08 \pm 20.80$ & $0.0000 \pm 0.0000$ & 0 \\
\addlinespace[2pt]

\multirow{2}{*}{$1\%$}
& Off & 4,320 & $73.08 \pm 36.35$ & $ 0.00 \pm 0.00$  & $34.35 \pm 14.31$ & $0.0000 \pm 0.0000$ & 0 \\
\addlinespace[2pt]

\multirow{2}{*}{$1\%$}& On  & 4,320 & $88.11 \pm 24.76$ & $56.18 \pm 15.83$ & $24.73 \pm 17.78$ & $0.0000 \pm 0.0000$ & 0 \\
\addlinespace[2pt]

\multirow{2}{*}{$5\%$}
& Off & 4,320 & $53.41 \pm 36.22$ & $ 0.00 \pm 0.00$  & $44.88 \pm 7.91$  & $0.0000 \pm 0.0000$ & 0 \\
\addlinespace[2pt]

\multirow{2}{*}{$5\%$}& On  & 4,320 & $78.31 \pm 30.74$ & $51.96 \pm 20.40$ & $36.13 \pm 11.69$ & $0.0000 \pm 0.0000$ & 0 \\
\bottomrule
\end{tabular}
\begin{tablenotes}
\footnotesize
\item \textit{Notes:} ECC denotes the optional erasure-recovery mechanism used by the receiver. Acceptance measures the proportion of epochs accepted after freshness, window, and MAC checks; recovery measures epochs reconstructed after partial channel loss; and re-sync success measures successful restoration of receiver alignment after incomplete or delayed epochs.
\end{tablenotes}
\end{threeparttable}
\end{table*}

\begin{table*}[t]
\centering
\caption{QIM-style multichannel robustness under channel jitter. Results are reported as mean~$\pm$~standard deviation over simulator runs.}
\label{tab:qim_resync_latency}
\setlength{\tabcolsep}{5pt}
\renewcommand{\arraystretch}{1.15}
\resizebox{0.86\textwidth}{!}{%
\begin{tabular}{@{}llrcccc@{}}
\toprule
\textbf{Jitter (ms)} & \textbf{ECC} & \textbf{Runs} & \textbf{Mean Latency (ms)} & \textbf{P95 Latency (ms)} & \textbf{Mean Re-sync Delay (ms)} & \textbf{Throughput (epochs/s)} \\
\midrule
$0$   & Off & 3,380 & $78.76 \pm 33.00$ & $129.95 \pm 60.33$ & $85.32 \pm 26.45$ & $6.66 \pm 3.81$ \\
$0$   & On  & 3,371 & $49.22 \pm 18.68$ & $123.67 \pm 56.82$ & $76.70 \pm 22.02$ & $8.77 \pm 2.51$ \\
\addlinespace[2pt]
$25$  & Off & 3,240 & $81.74 \pm 32.66$ & $139.30 \pm 57.43$ & $84.01 \pm 25.07$ & $7.12 \pm 3.59$ \\
$25$  & On  & 3,240 & $50.48 \pm 18.44$ & $128.50 \pm 54.77$ & $75.64 \pm 19.75$ & $8.85 \pm 2.45$ \\
\addlinespace[2pt]
$50$  & Off & 3,240 & $85.00 \pm 31.37$ & $152.91 \pm 54.67$ & $84.27 \pm 22.15$ & $7.11 \pm 3.59$ \\
$50$  & On  & 3,240 & $51.88 \pm 17.77$ & $136.60 \pm 51.56$ & $75.95 \pm 17.95$ & $8.83 \pm 2.45$ \\
\addlinespace[2pt]
$100$ & Off & 3,240 & $95.12 \pm 28.51$ & $184.94 \pm 52.41$ & $87.72 \pm 18.54$ & $6.96 \pm 3.62$ \\
$100$ & On  & 3,240 & $56.76 \pm 15.95$ & $159.27 \pm 47.11$ & $78.18 \pm 15.40$ & $8.72 \pm 2.55$ \\
\bottomrule
\end{tabular}%
}
\vspace{2pt}
\begin{minipage}{0.98\textwidth}
\footnotesize
\emph{Note:} ECC denotes the optional erasure-recovery mechanism used by the receiver. Mean latency and P95 latency measure end-to-end epoch acceptance delay, while mean re-synchronisation delay measures the time required to restore receiver alignment after incomplete, delayed, or reordered channel observations.
\end{minipage}
\end{table*}





\subsection{Masking Ablations}
\label{subsec:masking_ablations}

We evaluated the Ours-QIM masking layer separately from image embedding by
varying the public cover feature $\psi(\cdot)$, the relative QIM step
$\Delta/\Delta_0\in\{0.5,1,2\}$, and the dither source. For the all-zero
message, $b=\mu$, so the experiment directly measures the generated mask.
Across the $33$ configurations and payloads of $0.10$, $0.20$, and
$0.40$~bpp, the absolute bit bias remained within
$[0.0011,0.0027]$, empirical min-entropy remained within
$[0.9922,0.9968]$ bits per output bit, and monobit/runs pass rates were
$98$--$100\%$ and $96$--$100\%$, respectively. The small differences among
feature maps, QIM steps, and dither controls indicate that the observed mask
statistics are stable under the tested parameter changes.

A leave-one-seed-out logistic warden was also trained to distinguish masked
payload windows from independent uniform controls. Over all $165$ settings,
mean accuracy was $0.5004$, mean AUC was $0.4993$, and the mean normalised
proxy
\begin{equation}
 \widehat{\varepsilon}_{\mathsf{param}}^{\mathsf{clf}}
 =
 2\left|\widehat a-\frac12\right|
\end{equation}
was $0.0291$ (maximum $0.0933$). For the reference configuration
($\psi=\textsf{HashID}$, cover-conditioned PRF dither, and
$\Delta=\Delta_0$), the corresponding mean accuracy, AUC, and proxy were
$0.5058$, $0.5037$, and $0.0297$. These near-chance results provide
classifier-dependent empirical support for a small masking contribution
associated with $\varepsilon_{\mathsf{param}}$ in the composition theorem;
they are not formal estimates or upper bounds on total variation.

The existing deployed configuration additionally achieved mean BER
$1.102\times10^{-3}$ and $99.71\%$ exact recovery over $10{,}000$ images at
$0.30$~bpp. This result supports practical extraction but is not claimed for
every ablation or payload. Full configuration-wise results, message-class
comparisons, recovery limitations, and the relationship to
$\varepsilon_{\mathsf{emb}}$ and
$\varepsilon_{\mathsf{sched}}(\sigma^\star)$ are reported in
Supplementary Material (see Section~G-B).




\section{Conclusion}
\label{sec:conclusion}

This paper introduced a cover-parameterised multichannel hybrid
steganographic framework that combines message-independent
cover-parameter synthesis with adaptive cover modification. The
framework formalises cover-parameter indistinguishability and
unlinkability, multichannel trace indistinguishability, and active
confidentiality and integrity under full cross-channel observation
and manipulation. Its protocol integrates cover-conditioned
QIM-style masking, QIM-Fused-CF embedding,
authenticated epoch binding, bounded scheduling, and
re-synchronisation. The empirical results indicate near-random
detectability at low payloads, but substantially greater exposure
at higher payloads, thereby identifying a practical operating
region rather than implying payload-independent security.
Overall, the findings demonstrate that synthesis, masking,
embedding, scheduling, authentication, and receiver state must be
jointly designed to achieve secure and operationally robust
multichannel steganographic communication.

\bibliographystyle{IEEEtran}
\bibliography{bare_arXiv}





\end{document}